%% file: driver_file.tex
\documentclass[10pt,conference,letterpaper]{IEEEtran}
\usepackage{times,amsmath,epsfig}
\usepackage{graphicx}
\usepackage{grffile}
\usepackage{tikz}
\usetikzlibrary{shapes,arrows,shapes.multipart}
\usepackage{algorithm}
\usepackage[noend]{algpseudocode}
\usepackage{xspace}
\usepackage{mathptmx}
\usepackage{caption}
\usepackage{subfigure}
\usepackage{balance}  
\usepackage{times}

\usepackage{amssymb}

\DeclareMathAlphabet{\mathcal}{OMS}{cmsy}{m}{n}
\DeclareMathAlphabet{\mathpzc}{OT1}{pzc}{m}{it}
\usepackage[T1]{fontenc}
\newcommand{\basicdp}{\textsc{BasicDP}\xspace}

\newcommand{\opt}{\textbf{OPT} }
 \newcommand{\todo}{\color{red} Comment: }  
\newcommand{\gm}{\textsf{GM}\xspace}
\newcommand{\wm}{\textsf{WM}\xspace}
\newcommand{\UM}{\textsf{UM}\xspace}  
\newcommand{\E}{\textsf{EM}\xspace}

\newcommand{\eat}[1]{}
\newcommand{\mech}{\mathcal{P}}
\newcommand{\para}[1]{\smallskip \noindent {\bf #1}}
\newcommand{\etal}{{\em et al.\xspace}}

\newcommand{\F}{\textsf{F}\xspace}
\newcommand{\0}{\ensuremath{\mathbb{L}_0}} 
\newcommand{\Ld}[1]{\ensuremath{\mathbb{L}_{0,#1}}} 
\newcommand{\1}{\ensuremath{\mathbb{L}_1}}
\newcommand{\2}{\ensuremath{\mathbb{L}_2}}

\newtheorem{definition}{Definition}
\newtheorem{theorem}{Theorem}
\newtheorem{example}{Example}
\newtheorem{lemma}{Lemma}

\newtheorem{fact}{Fact}

\DeclareMathOperator{\trace}{trace}
\newcommand{\yes}{Y}
\newcommand{\no}{N}

\usepackage{hyperref}
\let\labelindent\relax
\usepackage{enumitem}
\setlist{nosep}

\title{Constrained Differential Privacy for Count Data}
\author{%
{Graham Cormode{\small $~^{\#1}$}, Tejas Kulkarni{\small $~^{*1}$},Divesh Srivastava{\small $~^{\#2}$} }%
\vspace{1.6mm}\\
\fontsize{10}{10}\selectfont\itshape

\fontsize{9}{9}\selectfont\ttfamily\upshape
$^{\#1}$\,g.cormode@warwick.ac.uk\\
$^{*1}$\,tejasvijaykulkarni@gmail.com%
\vspace{1.2mm}\\
\fontsize{10}{10}\selectfont\rmfamily\itshape
$^{\#1},^{*1}$\,The University Of Warwick, UK\\
\fontsize{9}{9}\selectfont\ttfamily\upshape
$^{\#2}$\,divesh@research.att.com \\
\fontsize{10}{10}\selectfont\rmfamily\itshape
$^{\#2}$\,AT\&T Labs-Research, USA\\
}
\begin{document}
\maketitle
\begin{abstract}
Concern about how to aggregate sensitive user data without compromising
individual privacy is a major barrier to greater availability of
data.
The model of differential privacy has emerged as an accepted model to
release sensitive information while giving a statistical guarantee for
privacy.
Many different algorithms are possible to address different target
functions.
We focus on the core problem of count queries, and seek to design
{\em mechanisms} to release data associated with a group of $n$
individuals. 	
 
Prior work has focused on designing mechanisms by raw optimization of a
loss function, without regard to the consequences on the results. 
This can leads to mechanisms with undesirable properties, such as
never reporting some outputs (gaps), and overreporting others
(spikes). 
We tame these pathological behaviors by introducing a set of desirable properties that
mechanisms can obey.
Any combination of these can be satisfied by solving a linear program (LP)
which minimizes a cost function, with constraints enforcing the properties.
We focus on a particular cost function, and 
provide explicit constructions that are optimal for certain
combinations of properties, and show a closed form for their cost.
In the end, there are only a handful of distinct optimal mechanisms to choose
between: one is the well-known (truncated) geometric mechanism; the
second a novel mechanism that we introduce here, and the remainder are
found as the solution to particular LPs.
These all avoid the bad behaviors we identify. 
We demonstrate in a set of experiments on real and
synthetic data which is preferable in practice, for different
combinations of data distributions, constraints, and privacy parameters. 
\end{abstract}
\input{intro}

\input{model}

\section{Unconstrained Mechanism Design}
\label{sec:unconstrained}


A natural starting point is to use optimization tools to find optimal
mechanisms.
Following~\cite{Ghosh:2009}, 
the key observation is that the DP requirements can be written as
linear constraints over variables which represent the entries of the
mechanism.
The objective function is also a linear function of these variables.  
Formally, we define variables $\rho_{i,j}$ for
$\Pr[i|j]$, and write: 
\begin{align}
  \label{eq:lpobj}
  \text{minimize:} \quad \sum_{j=0}^{n}w_j\sum_{i=0}^{n} |i-j|^{p}\rho_{i,j} \\
  \text{subject to:}  \quad \label{eq:nonzero}
 0 \le \rho_{i,j} \le 1 \quad \forall i,j \in [n] \\
  \label{eq:col} 
 \sum_{i=0}^{n} \rho_{i,j} =1 \quad \forall j \in [n] \\
    \label{eq:dp1} 
    \rho_{i,j} \geq \alpha \rho_{i,j+1}, \text{ and } 
    \rho_{i,j+1} \geq \alpha \rho_{i,j} \quad \forall i \in [n], j \in [n-1]
\end{align}

The constraints can be understood as follows:
\eqref{eq:nonzero}, \eqref{eq:col} ensure that the entries of the matrix are
probabilities and each column encodes a
probability distribution, i.e. sums to $1$.  Constraint \eqref{eq:dp1}
encodes the differential privacy constraints.
Finally, \eqref{eq:lpobj} encodes a loss function of Definition~\ref{eq:lpobj}  for the notion of utility we aim for. 
We refer to the set of 
constraints~\eqref{eq:nonzero}, \eqref{eq:col} and \eqref{eq:dp1}  as
\basicdp.
%
%
The result is a linear program with a quadratic number of variables,
and a quadratic number of constraints, each containing at most a
linear number of variables.
Therefore, solving the resulting LP obtains a 
mechanism minimizing the given objective function with the desired
properties, in time polynomial in $n$.




Applying this approach yields results like those in
Figure~\ref{fig:anomalies}.
Our studies found that similar undesirable results were found across a
range of choices of $n$, $\alpha$ and loss function.
Simple attempts to prevent these outcomes are not effective.
For example, we can ensure that no entry is zero by adding a
constraint to the LP enforcing this.
However, the consequence is that rows which were zero are now set to
be the smallest allowable value, which is unsatisfying.  
Instead, we propose an additional set of properties to eliminate
degeneracy and provide more structure in our solutions.

\section{Constrained Mechanism Design}
\label{sec:smallgroup}


\subsection{Structural Constraints}
\label{sec:properties}
We now propose a set of structural properties that
help to control the objective function
in addition to meeting differential privacy. 
We believe that these constraints are natural and intuitive 
and often observed in other mechanisms satisfying differential privacy. 
We present properties of three types: those which
operate on rows of the matrix, those which apply to columns of the
matrix, and those which apply to the diagonal. 


\para{Row Honesty (RH):} A mechanism is {\em row honest} if 
\begin{equation}
  \forall i, j. \Pr[i|i] \geq \Pr[i|j]
\label{eq:rowhonest}
\end{equation}
\noindent
Row honesty means that a mechanism should
have higher probability of reporting $i$ when the input is $i$
than for any other input. 

\para{Row Monotone (RM):}
A mechanism is {\em row monotone} if
\begin{align}\nonumber
  \forall 1 \leq j \leq i: \Pr[i|j-1] \leq \Pr[i|j] \\
  \forall i \leq j < n: \Pr[i|j+1] \leq \Pr[i|j]
\label{eq:rowmonotone}
\end{align}

This property generalizes row honesty: row monotonicity implies row
honesty.
It requires that entries in row $i$ are monotone non-increasing as
we move away from the diagonal element $\Pr[i|i]$. 
Note that row monotonicity is independent of differential privacy: we
can find mechanisms that achieve DP but are
not row monotone, and vice-versa. 

Analogous to the row-wise properties, we define monotonicity and
honesty along columns also.

\para{Column Honesty (CH):} A mechanism is {\em column honest} if
\begin{equation}
\forall i, j: 
\Pr[j|j] \geq \Pr[i|j].
\label{eq:columnhonesty}
\end{equation}

Column honesty requires that the mechanism
be \textit{honest} enough
to report the true answer more often than any individual false answer.  As
demonstrated by Example~\ref{eg:gm}, \gm does not obey column honesty.  

\noindent
\mbox{\para{Column Monotone (CM)}: A mechanism is {\em column
    monotone} if}
\begin{align}\nonumber
  \forall 1 \leq i \leq j: \Pr[i-1|j] \leq \Pr[i|j] \\
  \forall j \leq i < n: \Pr[i+1|j] \leq \Pr[i|j]
  \label{eq:columnmonotone}
\end{align}

As in the row-wise case, column monotonicity implies column honesty
(but not vice-versa). 
It captures the property that outputs closer to the true answer should
be more likely than those further away.

\eat{
A third column property emerges by analogy with the differential
privacy requirement (Definition~\ref{def:dp}).
We can apply the same relationship column-wise: 

\para{(Column) Publisher DP (CDP):} A mechanism respects {\em publisher DP}
with regard to a parameter $\alpha' \in [0,1]$
if
\begin{equation}
  \forall i,j: \alpha' \leq \frac{\Pr[i|j]}{\Pr[i+1|j]} \leq
  \frac1{\alpha'}.
  \label{eq:pubdp}
\end{equation}

This property captures the constraint that the ratio of probabilities
between neighboring {\em outputs} is bounded, compared to the usual DP
definition, which bounds the ratio for neighboring inputs. 
}

\para{Fairness (F):} A mechanism is {\em fair} when the probability of
reporting the true input is constant, i.e.
\begin{equation}
  \forall i, j : \Pr[i|i]=\Pr[j|j] := y.
  \label{eq:fair}
\end{equation}
Example~\ref{eg:gm} shows that \gm is not a fair mechanism. 
If a mechanism is fair and has row honesty, then all
off-diagonal elements are at most $y$, so the mechanism also satisfies
column honesty.
Symmetrically, a fair and column honest mechanism is row honest. 
While this may seem like a restrictive constraint, we observe that
mechanisms proposed in other contexts have this property, such as the
staircase mechanism of~\cite{GKOP:15}. 



\begin{lemma}
If a mechanism is required to be fair, then any mechanism that
minimizes the objective $O_{0,\sum}$ is simultaneously optimal for all
settings of weights $w_j$. 
\label{lemma:faircost}
\end{lemma}

\begin{IEEEproof}
  Let the diagonal element of the fair mechanism be $y$.
  The objective function value is
  \begin{equation}
\sum_{j \in [n]} \sum_{i \in [n]} w_j \Pr[i|j] (i-j)^{0}=\sum_{j \in [n]} w_j (1-y)=1-y
  \end{equation}

That is, the value is independent of the $w_j$s. 
\end{IEEEproof}

\noindent
\mbox{\para{Weak Honesty (WH):} A mechanism satisfies {\em weak honesty} if}
\begin{equation}
  \forall i: \Pr[i|i] \geq \frac{1}{n+1}
  \label{eq:weakhonest}
\end{equation}
We can consider this property a weaker version of column honesty, as
CH implies WH:
for any column $j$, summing the column honesty property over all rows $i$, we obtain
\[ (n+1) \Pr[i|i] = \sum_{i=0}^n \Pr[j | j] \geq \sum_{i=0}^n \Pr[i |
  j] = 1 \]
so after rearranging, we have $\Pr[i|i] \geq \frac{1}{n+1}$. 

Weak honesty ensures that a mechanism reports the true answer with
probability at least that of uniform guessing (formalized as 
the uniform mechanism \UM in Definition~\ref{def:uniform}). 
It also ensures that the mechanism does not have any rows that are all
zero (corresponding to outputs with no probability of being produced). 
\gm does not always obey weak honesty, as is shown by Example~\ref{eg:gm}. 

 

The final property we consider is a natural symmetry property
(formally, it is that the matrix $\mech$ is {\em centrosymmetric}):

\eat{
\begin{definition} 
  (\textit{Enforceable Property}) A property is enforceable if it can be explicitly materialized as a linear programming constraint (or in any other algorithmic way) without violating other constraints for all sizes of dataset and privacy parameter.
 \end{definition}
 We have already hinted about symmetry before but we define it again for completeness. \\
}

\para{Symmetry (S):} A mechanism is {\em symmetric} if
\begin{equation}
  \forall i, j: \Pr[i|j]=\Pr[n-i|n-j]
  \label{eq:symmetry}
\end{equation}

Since the input and output domains, and the objective functions are
symmetric, it is natural to seek mechanisms which are also symmetric.
Our next result shows that symmetry is always achievable without any
loss in objective function.

\begin{theorem}
Given a mechanism $M$ which meets a subset of properties $P$ from
those defined above, we can construct a symmetric mechanism $M^*$ which also
satisfies all of $P$ and achieves the same objective function value as $M$. 
\label{thm:symm_mechanism}
\end{theorem}

We defer this proof to the Appendix. 

\para{Consequences of these properties.}
We first argue that these properties all contribute to avoiding the
degenerate mechanisms shown above.
The (column, row) honesty and monotonicity properties work to prevent
the ``spikes'' observed when a value far from the true input is made
excessively likely.
The (column) honesty properties do so by preventing a far output being more
likely than the true input; the (column) monotonicity properties do so more
strongly by ensuring that any further output is no more likely than one
that is nearer to the true input.
Fairness, column honesty and weak honesty prevent gaps (zero rows): they ensure that
the diagonal entry in each row is non-zero, and then the DP
requirement ensures that all other entries in the same row must also
be non-zero.

A simple observation is that in the (trivial) $n=1$ case, randomized
response is the unique optimal $\alpha$-differentially private
mechanism under any objective function $O_{p,\sum}$.
It is straightforward to check that the resulting mechanism meets all
the above properties when $p\geq \frac12$. 
We next show that there is an efficient procedure to find an optimal
constrained mechanism for any $n>1$.

\begin{theorem}
Given any subset of the structural constraints, we can find an optimal
(constrained) mechanism which respects these constraints in time
polynomial in $n$.
\end{theorem}

\begin{IEEEproof}
  We break the proof into two pieces.  First, we argue that given any
  subset of structural constraints we can create a Linear Program
  describing it, and second we argue that there exists a mechanism
  satisfying them all.
Observe that all seven properties listed above  can
be encoded as a linear constraints.
For example, symmetry is written as
\[ \rho_{i,j} = \rho_{n-i, n-j} \quad \forall i, j \in [n]\]

while weak honesty is
\[ \rho_{i,i} \ge 1/(n+1). \]
Row monotonicity becomes
\[
\begin{array}{rl}
  \rho_{j,i-1} \leq \rho_{j,i} \quad &\forall j \in [n], i < j \\
  \rho_{j,i+1} \leq \rho_{j,i} \quad &\forall i \in [n-1], j < i
\end{array}
\]
  
Consequently, we can create a linear program of size polynomial in
$n$, by adding these to the \basicdp constraints~\eqref{eq:nonzero},
\eqref{eq:col} and \eqref{eq:dp1} established in
Section~\ref{sec:unconstrained}.
This shows the first part of the proof. 
Next, we show that any such LP is feasible by defining a trivial
baseline mechanism:

\begin{definition}[Uniform Mechanism, \UM]
\label{def:uniform}
  The {\em uniform me\-ch\-anism} of size $n$ has
  $\Pr[i|j] = \frac1{n+1}$, for all $i, j \in [n]$.
\end{definition}

That is, \UM ignores its input and picks an allowable output uniformly
at random.
It demonstrates that all our properties are (simultaneously)
achievable, albeit trivially.
By observation, the mechanism is symmetric and fair for any $\alpha' \le 1$.
It meets the inequalities specified for row monotonicity, column
 monotonicity and weak honesty with equality.  
\UM also satisfies differential privacy for all $\alpha \le 1$.
\end{IEEEproof}

Clearly, \UM is undesirable from the perspective of providing
utility.
We easily calculate that the objective function value $O_{0,\sum}$ achieved by
\UM is $\frac{n}{n+1}$, which is close to the maximum possible value of
1.
Note that we chose our definition of the \0 function to assign this mechanism a
(reweighted) score of~1.

\begin{figure}[t]
\eat{
  \begin{minipage}{.4\linewidth}
     $\begin{pmatrix}
  \mathbf{x} & \Rightarrow  & \Rightarrow &\Rightarrow& \cdots &\Rightarrow \\
 \Leftarrow & \mathbf{y} & \Rightarrow  & \Rightarrow &\cdots &  \Rightarrow \\
 \Leftarrow & \Leftarrow & \mathbf{y} & \Rightarrow &\cdots & \Rightarrow\\
  \Leftarrow & \Leftarrow & \Leftarrow & \mathbf{y}  & \cdots &\Rightarrow \\
   \Leftarrow &\Leftarrow & \Leftarrow & \Leftarrow  & \cdots &\Rightarrow \\

  \vdots  & \vdots  & \vdots & \vdots & \ddots & \vdots \\
  \Leftarrow & \Leftarrow & \Leftarrow & \Leftarrow & \cdots &	\mathbf{x} 
 \end{pmatrix}$ \\ 
     \centerline{a) unimodal signature}
  \end{minipage}
  \begin{minipage}{.6\linewidth}
}
\centering
$ \begin{pmatrix}
  {x} & x\alpha  & x\alpha^{2}&x\alpha^{3}& \cdots &x\alpha^{n} \\
 y \alpha & {y} & y \alpha  & y\alpha^{2} &\cdots &  y\alpha^{n-1} \\
 y \alpha^{2} & y\alpha & {y} & y\alpha &\cdots & y\alpha^{n-2} \\
  y \alpha^{3} & y\alpha^{2} & y\alpha & {y}  & \cdots &y\alpha^{n-3} \\
   y \alpha^{4} & y\alpha^{3} & y\alpha^{2} & y\alpha  & \cdots &y\alpha^{n-4} \\

  \vdots  & \vdots  & \vdots & \vdots & \ddots & \vdots \\
  x\alpha^{n} & x\alpha^{n-1} & x\alpha^{n-2} & x\alpha^{n-3} & \cdots &	{x} 
\end{pmatrix}$
  \caption{Structure of \gm, where $x=\frac{1}{1+\alpha}$ and $y=\frac{1-\alpha}{1+\alpha}$ }
\label{fig:gm}
\end{figure}						

\subsection{The Geometric Mechanism}

Next, we revisit the (range restricted) Geometric Mechanism, \gm
(Definition~\ref{def:gm}).
In Figure~\ref{fig:gm}, we show the structure of the mechanism, which
can be derived by simple calculation from Definition~\ref{def:gm}. 
Below, we show that it enjoys a number of special properties.
In prior work, Ghosh \etal\ showed that \gm plays an important role, as
it can be transformed into an optimal mechanism for different
objectives.
Here, we argue
(proof omitted due to space constraints)
a more direct result: that \gm is directly optimal for a
uniform objective function\footnote{Note that, compared to \protect\cite{Ghosh:2009}, 
we define mechanisms to enforce differential privacy along rows of $\mech$
rather than columns.}

\begin{theorem}
\label{thm:gm_dpbasic}
\gm is the (unique) optimal mechanism satisfying \basicdp under the \0
objective function. 
\label{gm_basicdp}
\end{theorem}

\para{Limitations of \gm.}
Since \gm is `optimal' for \0, should we conclude our study here?
The answer is no, since \gm fails to satisfy many of the desirable
properties we identified in Section~\ref{sec:properties}, and as
illustrated in Example~\ref{eg:gm}.
We have already observed that \gm is not fair, and does not in general
satisfy column honesty (or column monotonicity) or weak honesty.
{Next, we identify parameter settings for when they do hold.}
  
 \begin{lemma}
 \label{lemma:uniformity}
\gm obeys weak honesty iff $n \geq \frac{2\alpha}{1-\alpha}$. 
 \end{lemma}

 \begin{IEEEproof}
Weak honesty requires the diagonal elements to all exceed
$\frac{1}{n+1}$. 
Since $y < x$, we focus on $y$. 
We require $y \geq \frac{1}{n+1}$ i.e $\frac{1-\alpha}{1+\alpha} \geq
\frac{1}{n+1}.$
This reduces to $n+1 \geq \frac{1+\alpha}{1-\alpha}$, giving the
requirement $n \geq \frac{2\alpha}{1-\alpha}$.
 \end{IEEEproof} 

 \gm satisfies the column monotonicity condition for many $i, j$
 pairs. 
 The critical place in the matrix where it can be violated is 
between the first and second rows (symmetrically, between 
penultimate and final rows). 
This corresponds to the problematic behavior of \gm to report extreme
outputs (0 or $n$) disproportionately often. 
 \begin{lemma}
\label{lemma:cmonotone}
   \gm achieves column monotonicity iff $\alpha \leq \frac{1}{2}.$
\end{lemma}

\begin{IEEEproof}
We require $\Pr[1|1] \leq \Pr [0|1]$, i.e. $y \leq \alpha x$ or 
$\frac{1-\alpha}{1+\alpha} \leq \frac{\alpha}{1+\alpha}$.
This gives the condition $\alpha \leq \frac{1}{2}$.
It is straightforward to check that this ensures monotonicity in all
other columns. 
\end{IEEEproof}

By inspection, \gm is always symmetric, and row monotone.
The (\0) objective function value achieved by \gm is
\begin{align*}
  \textstyle
\frac{n+1}{n}\big(
  1- \frac{(n-1)y + 2x}{n+1}\big)=
  \frac{n+1}{n}\big(1- \frac{n-1}{n+1}\frac{1-\alpha}{1+\alpha}-\frac{2}{(1+\alpha)(n+1)}\big)=
\frac{2\alpha}{1+\alpha}
\end{align*}

We next design a different
explicit mechanism which achieves more of the desired properties.

\subsection{Explicit Fair Mechanism}
\label{sec:fair}

Although we can achieve any desired combination of properties by
solving an appropriate linear program, it is natural to ask whether
there is any non-trivial explicit mechanism that achieves properties
such as fairness with an objective function score comparable to that
of \gm.
We answer this question in the positive. 
First, we consider the limits of what can be achieved under
fairness. 
In the case of \gm, all DP inequalities are tight.
This is not possible when fairness is demanded.
A fair mechanism $M$ with all DP inequalities tight would be completely
determined: $M_{i,j} = y \alpha^{|i-j|}$ for some $y$.
It is easy to calculate for any such mechanism that there is no
setting of $y$ which ensures that all columns sum to 1, a contradiction.
Hence, we cannot have a fair mechanism with all DP inequalities tight. 
Nevertheless, trying to achieve tightness provides us with a 
bound on what can be achieved.

\begin{lemma}
 \label{lemma:fairness_constraint}
Let \F  be a fair mechanism of size $(n+1) \times
(n+1)$ with $y$ as the diagonal element.
Then
$y \leq \frac{1-\alpha}{1+\alpha-2\alpha^{\frac{n}{2}+1}}$.
\end{lemma}

\begin{IEEEproof}
There are some slight differences depending on whether we consider odd
or even values of $n$.
Without loss of generality, take $n$ even.
We will consider a fixed column $j$. 
For all $i$, we are required to have $\Pr[i|i] = y$ for some $y$.
Repeatedly applying the DP inequality, we obtain an upper bound
involving $y$ as 
$\Pr[i|j] \geq y\alpha^{i-j}$ when $j <i$ and $\Pr[i|j] \geq
y\alpha^{j-i}$ when $i >j$. 
Summing these for any given column $j$ and equating to 1 provides an
upper bound on $y$.
We get the tightest bound by picking column
$j=\frac{n}{2}$.  Then $y + 2y \sum_{j=1}^{\frac{n}{2}} \alpha^{j}\leq
1$, so:
%
\begin{equation}
  y \leq \frac{1}{1+2\sum_{j=1}^{\frac{n}{2}} \alpha^{j}}
=\frac{1-\alpha}{1+\alpha-2\alpha^{\frac{n}{2}+1}}
\label{eq:fairvalue}
\end{equation}
For $n$ large enough, we can neglect the
$\alpha^{n/2 + 1}$ term, and approximate this quantity
by $\frac{1-\alpha}{1+\alpha}$.
\end{IEEEproof}

\begin{figure}[t]
\centering
  $ \begin{pmatrix}
y & y\alpha & y\alpha^{2} & y\alpha^{3} & y\alpha^{4} & y\alpha^{4} & y\alpha^{4} & y\alpha^{4}  \\
y\alpha & y & y\alpha & y\alpha^{2} & y\alpha^{3} & y\alpha^{3} & y\alpha^{3} & y\alpha^{3}  \\
y\alpha & y\alpha & y & y\alpha & y\alpha^{2} & y\alpha^{3} & y\alpha^{3} & y\alpha^{3}  \\
y\alpha^{2} & y\alpha^{2} & y\alpha & y & y\alpha & y\alpha^{2} & y\alpha^{2} & y\alpha^{2}  \\
y\alpha^{2} & y\alpha^{2} & y\alpha^{2} & y\alpha & y & y\alpha & y\alpha^{2} & y\alpha^{2}  \\
y\alpha^{3} & y\alpha^{3} & y\alpha^{3} & y\alpha^{2} & y\alpha & y & y\alpha & y\alpha  \\
y\alpha^{3} & y\alpha^{3} & y\alpha^{3} &y\alpha^{3} & y\alpha^{2} & y\alpha & y & y\alpha  \\
y\alpha^{4} & y\alpha^{4} & y\alpha^{4} & y\alpha^{4} & y\alpha^{3} & y\alpha^{2} & y\alpha & y  \\
 \end{pmatrix}$

\eat{
   \begin{minipage}{.2\linewidth}
$\begin{pmatrix}
y & y\alpha & y\alpha^{2} & y\alpha^{3} & y\alpha^{3} & y\alpha^{3} & y\alpha^{3}  \\
y\alpha & y & y\alpha & y\alpha^{2} & y\alpha^{3} & y\alpha^{3} & y\alpha^{3} \\
y\alpha & y\alpha & y & y\alpha & y\alpha^{2} & y\alpha^{2} & y\alpha^{2}  \\
y\alpha^{2} & y\alpha^{2} & y\alpha & y & y\alpha & y\alpha^{2} & y\alpha^{2}  \\
y\alpha^{2} & y\alpha^{2} & y\alpha^{2} & y\alpha & y & y\alpha & y\alpha  \\
y\alpha^{3} & y\alpha^{3} & y\alpha^{3} & y\alpha^{2} & y\alpha & y & y\alpha  \\
y\alpha^{3} & y\alpha^{3} & y\alpha^{3} & y\alpha^{3} & y\alpha^{2} & y\alpha & y  \\
  \end{pmatrix}$ \\ 
  \end{minipage} }
\caption{Explicit fair mechanism for $n=7$}
\label{fig:explicit_mechanism}
\end{figure}

\eat{
\begin{lemma}
\label{lemma:fair_wHon_row/col_monotone}
A mechanism with fairness, weak honesty and row (column) monotone
enforced, satisfies column (row) honesty at no extra cost.
\end{lemma}
\begin{IEEEproof}

{\todo clarify this result, move earlie}
  
Consider a case of $i \leq j$. Let $y$ be the diagonal element. Wlog, for any row $i$, row monotone makes certain that $y\geq \Pr[i|j-1]\geq \Pr[i|j-2] \geq ... \geq \Pr[i|0]$. Similarly, for row $i+1$, $y\geq \Pr[i+1|j]\geq \Pr[i+1|j-1] \geq ... \geq \Pr[i+1|0]$. For mean for each column $j$, $y\geq \Pr[i+1|j]\geq \Pr[i+2|j]\geq ...\geq \Pr[n|j]$. Argument is symmetric for $i>j$ and for the case when row monotone is exhibited by a mechanism satisfying fairness, weak honesty and column monotone.
\end{IEEEproof}
}

Note that for optimality under an objective function $O_{p,\sum}$, we should
make $y$ as large as possible.
Hence, any optimal mechanism will have $y$ as close to this value as
possible. 
Indeed, the above proof helps us to design an explicit mechanism \E that achieves
fairness.
The proof argues that in column $n/2$, the smallest values we can
obtain above and below the $y$ entry are $\alpha y$, $\alpha^2 y$ and
so on up to $\alpha^{n/2}y$.
Then the sum of these terms is set to 1.
All other columns must also sum to 1; a simple way to achieve this is
to ensure all columns contain a permutation of the same set of terms. 
To ensure DP is satisfied, we should arrange these so that
row-adjacent entries differ in their power of $\alpha$ by at most
one.	

Our explicit fair mechanism \E is then defined as follows: 
\begin{equation}
\Pr[i|j] =
\begin{cases}
  y\alpha^{|i-j|} & \text{if } |i-j| < \min(j, n-j) \\
  y\alpha^{\lceil \frac{|i-j| + \min(j, n-j)}{2}\rceil} & \text{otherwise}
\end{cases}
\label{eq:fairmechanism}
\end{equation}

Here, $y$ is set to $\frac{1-\alpha}{1 + \alpha - 2\alpha^{n/2+1}}$,
  i.e. the value determined in Equation~\eqref{eq:fairvalue}.
From the proof of Lemma~\ref{lemma:faircost} and \eqref{eq:l0cost}, we have that the \0 score of
this mechanism is $\frac{n+1}{n}(1 - y)$, as it maximizes $y$ subject
to the bound of Lemma~\ref{lemma:faircost}. 

Figure~\ref{fig:explicit_mechanism} shows the instantiation of this
mechanism for the case $n=7$. 
Comparing to \gm, we see that the diagonal elements are slightly
increased, with the exception of the two corner diagonals, which are
decreased.
It is tempting to try to obtain the mechanism via the Exponential
Mechanism, by using a quality function applied to $|i-j|$ similar in
form to \eqref{eq:fairmechanism}.
Note however, that the constant factors of $2$ in its definition
\eqref{eq:em} leads to a considerably weaker result than this explicit
construction, equivalent to halving the privacy parameter $\epsilon$. 
It is easy to check that in the $n=7$ example, the mechanism is symmetric,
and meets all of the properties defined in
Section~\ref{sec:properties}.
In fact, this is the case for all values of $n$. 
The proof is rather lengthy and proceeds by considering a number of
cases; we omit it for brevity. 


\begin{theorem}
  \E is an optimal mechanism under \0 that
  satisfies all properties listed in Section~\ref{sec:properties}. 
\label{thm:fair}
\end{theorem}

\subsection{Comparing mechanisms}
\label{sec:comparing}

In Section~\ref{sec:properties}, we define 7 different properties,
denoted as RH, RM, CH, CM, WH, F, and S.
We can seek a mechanism that satisfies any subset of these, suggesting
that there are 128 combinations to explore.
However, we are able to dramatically reduce this design space with the
following analysis based on the \0 score function.

\begin{figure}[t]
\centering
\tikzset{
decision/.style = {diamond, draw, fill=green!70, 
  text width=5em, text badly centered,  inner sep=0pt},
block/.style = {rectangle, draw, fill=red!30, 
  text width=5em, text centered, rounded corners, minimum height=4em},
line/.style = {draw, -latex'},
cloud/.style = {draw, ellipse,fill=red!30, 
  minimum height=2em},
subroutine/.style = {draw,rectangle split, rectangle split horizontal,
  rectangle split parts=3,minimum height=1cm,
  rectangle split part fill={red!50, green!50, blue!20, yellow!50}},
connector/.style = {draw,circlefill=yellow!20},
data/.style = {draw, trapezium,fill=olive!20}
}
\resizebox{0.7\columnwidth}{!}{
\begin{tikzpicture}[scale = 0.7, auto]
\node [decision] (want_fairness) at (5,10) {Want Fairness?};
\node [block] (fair) at (2,8) {Fair Mechanism};
\node [decision, below of=want_fairness] (want_column) at (8,8) {Want Column Property?};
\node[decision] (n_alpha) at (3,5) {$n \ge
  \frac{2\alpha}{1-\alpha}$?}; 
\node [decision, below of=want_column, right of=n_alpha] (want_weak) at (4,4) {Want Weak Honesty?};
\node [block] (gm) at (2, 0) {Geometric Mechanism};
\node [block] (wm) at (7, 0) {WH};
\node [block, below of=want_column, right of=want_weak] (wmc) at (8, 4) {WH + CM};
%
\path [line] (want_fairness)--node {yes}(fair) ;
\path [line] (want_fairness) --node {no} (want_column);
\path [line] (want_column) --node {no} (n_alpha);
\path [line] (want_column) --node {yes} (wmc);
\path [line] (n_alpha) --node {yes} (gm);
\path [line] (n_alpha) --node {no} (want_weak);
\path [line] (want_weak) --node {no} (gm);
\path [line] (want_weak) --node {yes} (wm);

\end{tikzpicture}
}
\caption{Flowchart of properties for \0 objective  ($\alpha > \frac12$)}
\label{fig:flowchart}
\eat{\begin{tikzpicture}[scale = 0.7, auto]
\node [decision] (want_fairness) at (5,7) {Want Fairness?};
\node [block] (fair) at (2,5) {Fair Mechanism};
\node[decision, right of=want_fairness] (n_alpha) at (7,5) {$n \ge
  \frac{2\alpha}{1-\alpha}$?}; 
\node [decision, right of=want_fairness, below of=n_alpha] (want_weak) at (8,3) {Want Weak Honesty?};
\node [decision, below of=want_weak] (want_column) at (5,3) {Want Column Property?};
\node [block, left of=want_fairness, below of=want_column] (gm) at (4, 0) {Geometric Mechanism};
\node [block, right of=want_fairness, below of=want_column] (wm) at
(6, 0) {Weak Mechanism (via LP)};
\path [line] (want_fairness) --node {no} (n_alpha);
\path [line] (n_alpha) --node {yes} (want_column);
\path [line] (n_alpha) --node {no} (want_weak);        
\path [line] (want_fairness)--node {yes}(fair) ;
\path [line] (want_weak)--node {yes}(wm) ;
\path [line] (want_weak)--node {no}(want_column) ;
\path [line] (want_column)--node {no}(gm) ;
\path [line] (want_column)--node{yes} (wm) ;
\end{tikzpicture}
}
\end{figure}

\begin{figure}[t]
\centering
{\small
\begin{tabular}{c|cccc}
    Property & \gm & \wm & \E & \UM \\\hline
    Symmetry (S) & \yes & \yes & \yes & \yes\\
    Row Monotone (RM) & \yes & \yes & \yes & \yes\\
    Column Monotone (CM) & --- & --- & \yes & \yes \\
    Fairness (F) & \no & \no & \yes & \yes \\
    Weak Honesty (WH) & --- & \yes & \yes & \yes \\\hline
    \0 & $\frac{2\alpha}{1+\alpha}$
    & $\ge \frac{2\alpha}{1+\alpha}$
    & $\approx \frac{2\alpha}{1+\alpha}\cdot\frac{n+1}{n}$
    & $1$
\end{tabular}}
\caption{Properties of named mechanisms}
  \label{fig:table}
\end{figure}

First, we have shown by Theorem~\ref{thm:fair} that \E has the optimal
\0 score of any fair mechanism and has all other possible properties
``for free''.
Therefore, for any desired set of properties that include F, we can
just use \E.

Second, we have shown by Theorem~\ref{thm:gm_dpbasic} that \gm achieves
symmetry and row monotonicity (and hence row honesty) at a cost which
is optimal for any mechanism (i.e. \basicdp).
{Hence for any subset of $\{$S, RM, RH$\}$, it suffices to use \gm.}

In our experiments (Section~\ref{sec:exptswm}), we show that
there are only two remaining behaviors: either we solve the LP for the
WH property alone, or we solve the LP for WH and CM properties. 
Both solutions come with symmetry (S) and row properties RH, RM at no
additional cost. 
However, as noted in Lemma \ref{lemma:uniformity}, \gm satisfies WH
when $n \ge \frac{2\alpha}{1-\alpha}$, so in this case, we can use
\gm. 
Last, from observations in Section~\ref{sec:properties}, we have that
CM $\Rightarrow$ CH $\Rightarrow$ WH, so any demand that requires
any of these properties (and not F) can be satisfied by \wm also.
But in the weak privacy case that $\alpha \le \frac12$, \gm has these
properties, and so subsumes \wm.

To summarize this reasoning, in the case that $\alpha \le \frac12$,
there are only two competitive mechanisms: \E if fairness is required,
and \gm for all other cases. 
When $\alpha > \frac12$, things are a little more complicated, so we
show a flowchart in Figure~\ref{fig:flowchart}: from 128
possibilities, there are only four distinct approaches to consider
(two explicit mechanisms, and two solutions to an LP with different constraints),
and the choice is determined primarily by whether the mechanism is
required to satisfy fairness, column properties, weak honesty, or none. 
We also consider the baseline method \UM for comparison. 
We present a summary of these four named mechanisms in
Figure~\ref{fig:table}: the explicit \gm, \UM and \E, and \wm which is found by solving an LP. 
We write `---' for a property when this depends on the setting of the
parameters (discussed in the relevant section). 
We see that \E has a very similar objective function value \0
(recalling that we are trying to minimize this value), and all the
properties considered so far.
We do not have a closed form for the \0 score of \wm, as it is found
by solving the LP; however it is no less than that
for \gm (since \gm satisfies a subset of the required properties of
\wm), and no more than that of \E (since \E satisfies all properties).

At this point, we might ask how different are these mechanisms in
practice --- perhaps they are all rather similar?
Figure~\ref{fig:heatmaps} shows this is not the case for a small
group size ($n=4$).
For a moderate value of the privacy parameter $\alpha = 0.9$,
it presents the three non-trivial mechanisms using a
heatmap to highlight where the large entries are.  
We immediately see that \E concentrates probability mass along 
a uniform diagonal (as required by fairness).
Both \gm and \wm tend to favor extreme outputs (0 or 4 in this
example) whatever the input, although \gm is very skewed in this
regard while \wm is more uniform in allowing non-extreme outputs. 

Last, we check that what we are doing is not a trivial modification of
known mechanisms.
Prior work~\cite{Ghosh:2009,Gupte:2010} showed how optimal
unconstrained mechanisms can be derived from \gm by transformations.
Gupte and Sundarajrajan give a simple test: a mechanism $\mech$ can be
derived from \gm iff every set of three adjacent entries in
the mechanism satisfy
\smallskip\newline\centerline{$
(\Pr[i|j] - \alpha \Pr[i|j-1]) \ge \alpha (\Pr[i|j+1] - \alpha \Pr[i|j])
$}\smallskip\newline
We applied this test to mechanisms \wm and verified that this
condition is indeed violated for $n>1$.
For \E, this condition is automatically broken for all $n>1$: we have
$\Pr[2|0]$ $= $ $\Pr[2|1]$ $= y\alpha$, while
$\Pr[2|2]=y$. 
Then the condition is
\smallskip\newline\centerline{$
  y\alpha(1 - \alpha) \ge y\alpha(1 - \alpha^2) \equiv 1 \ge(1+\alpha)
$}\smallskip\newline
which is always false for $\alpha>0$. 
Hence, these mechanisms are not derivable from \gm. 

\begin{figure}[t]
\centering
\includegraphics[width=0.5\textwidth]{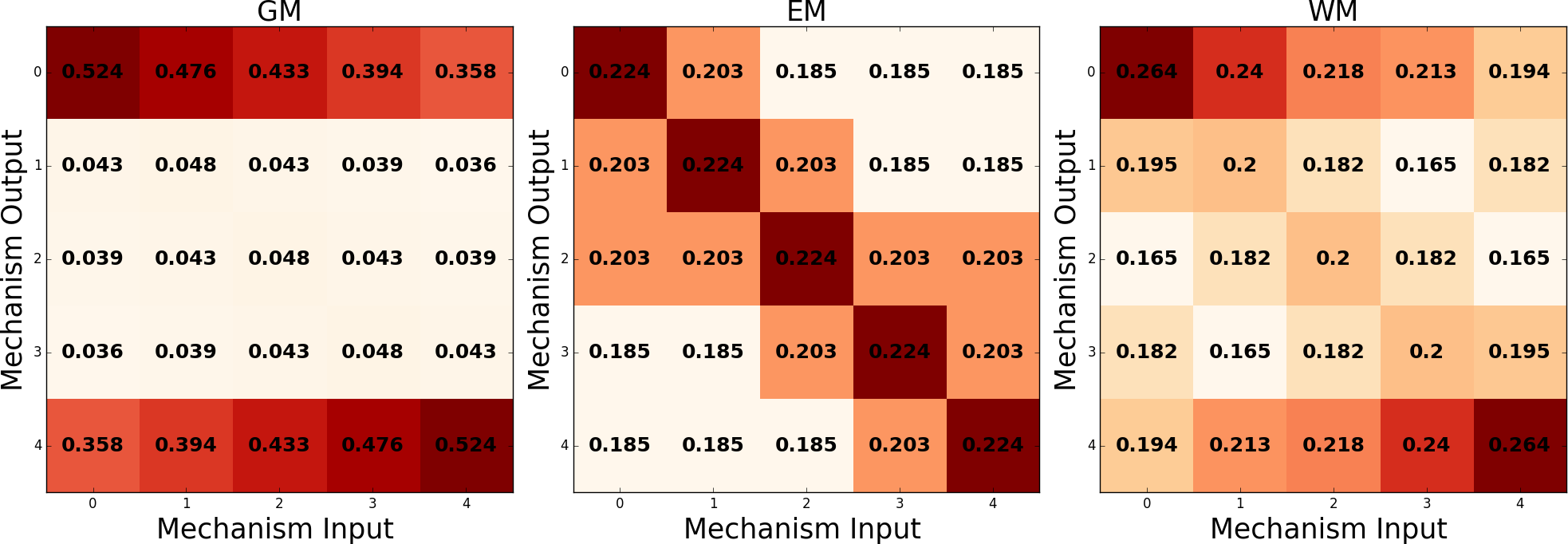}
  \caption{Heatmaps for \gm, \E, \wm with $n=4$}
  \label{fig:heatmaps}
\end{figure}

\input{expts_t}

\section{Concluding Remarks}
We have proposed and studied several structural properties for
privacy preserving mechanisms for count queries.
We show how any combination of desired properties can be provided
optimally under \0 by one of a few distinct mechanisms.
Our experiments show that the ``optimal'' \gm often displays the undesirable
property of tending to output extreme values (0 or $n$).
In practice, this means it is often not the mechanism of
choice, particuarly when $\alpha$ is large (above 0.7), but can be
acceptable for smaller privacy parameters.
\E and \wm are  quite different in structure, but are often
similar in performance. 	

It is natural to consider other possible properties---for example,
one could imagine taking a version of the DP constraint applied to
columns of the mechanism (in addition to the rows): this would 
enforce that the ratio of probabilities
between neighboring {\em outputs} is bounded, as well as that of
neighboring inputs. 
The next logical direction is to provide a deeper study of mechanisms
with various properties using \1 or \2 as objective function.
It will be interesting to study tailor-made linear programming
mechanisms that aim to optimize other queries such as range queries.

\bibliographystyle{IEEEtran}
\bibliography{papers}

\input{appendix}


\end{document}

%% file: intro.tex
\section{Introduction}
\allowdisplaybreaks

There has been considerable progress on the problem of how to release
sensitive information with privacy guarantees in recent years.
Various formulations have been proposed, with the model of
differential privacy emerging as the most popular and
robust~\cite{Dwork06}.
Differential privacy (DP) lays down rules on the likelihood on seeing
particular outputs given related inputs.
Many different algorithms have been proposed to meet this guarantee,
based on different objectives and input types~\cite{Dwork:Roth:14}. 
The resulting descriptions of output probabilities for different
inputs are referred to as mechanisms, such as the Laplace mechanism,
Geometric mechanism and Exponential mechanism described in more detail
subsequently.

In this paper, we focus on count queries,
a fundamental problem in private data release that underpins many
applications, from basic statistics of a dataset to complex spatial
and graphical distributions.
Count queries are needed to materialize frequency distributions,
instantiate statistical models, and as the basis of SQL COUNT *
queries.
Counts can be applied to arbitrary groups, and based on complex
predicates; hence they represent a very general tool. 
Abstracting, we have a group of $n$ individuals, who each hold a private bit
(encoding, for example, whether or not they possess a particular
sensitive characteristic).
The aim is to release information about the sum of the bits, while
meeting the stringent differential privacy guarantee.
The usual model assumes the existence of a trusted aggregator, who
receives the individual bits, and who aims to release a noisy
representation of their sum.
Since the value of the true answer is in $\{0\ldots n\}$, it is
natural to restrict the output of the mechanism to this range also, to
ensure downstream compatibility with subsequent data analysis expecting
integer counts in this range.
If we analyze how existing approaches to differential privacy handle
this case, we find there are weaknesses. 
We consider the most relevant example, mechanisms obtained via a linear programming framework.

\begin{figure}[t]
\centering
\includegraphics[width=0.5\textwidth]{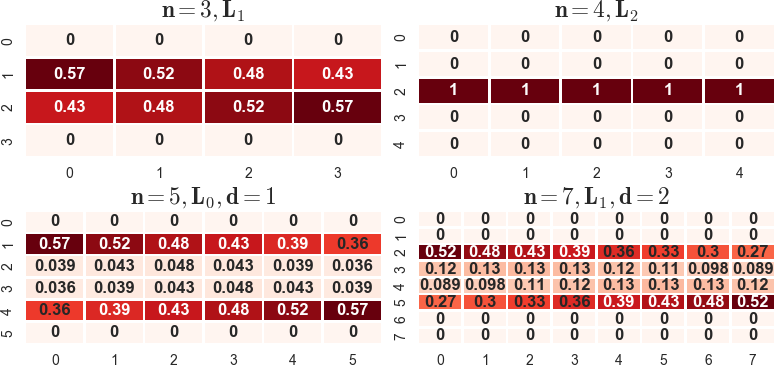}
  \caption{Heatmaps of unconstrained mechanisms for $\alpha=0.62$}
  \label{fig:anomalies}
\end{figure}

\para{Linear Programming Framework~\cite{Ghosh:2009}.}
Ghosh \etal\ considered count queries and proved powerful
theorems about utility-optimal mechanisms.
They showed how to design mechanisms for count queries which minimize a
loss function, via linear programming.
These mechanisms specify, for each possible input, a probability
distribution over allowable outputs. 
However, for common objectives, including to minimize 
the expected absolute error (denoted $\1$) and squared error ($\2$), 
we observed that the ``optimal'' mechanisms have some anomalous behavior, such
as never reporting some values. 

Figure~\ref{fig:anomalies} gives some examples of this phenomenon in
action.
We show four optimal mechanisms for different input sizes ($n$), under
a privacy guarantee controlled by a parameter $\alpha$ (explained
later, and set to a fixed value here).
Each column gives the probability distribution over the outputs in the
range $0$ to $n$, for a given input count (also $0$ to $n$).
The case of optimizing the squared error ($\2$) is most striking: the
``optimal'' thing to do in this case is to ignore the input and always report
`2'!
But other cases are also problematic: all these optimal mechanisms
never report some outputs (gaps), and disproportionately report some
others (spikes).
For example, minimizing the absolute error for $n=7$ has a chance of
reporting the values 2 or 5 with at least $0.7$ probability, regardless
of the input value.
Similarly, if we try to minimize the probability of reporting an
answer that is more than 1 step away from the true input (denoted as
$\0$ with $d=1$), there is an over 90\% chance of reporting 1 or 4.

\begin{figure}[t]
\centering
\includegraphics[width=0.5\textwidth]{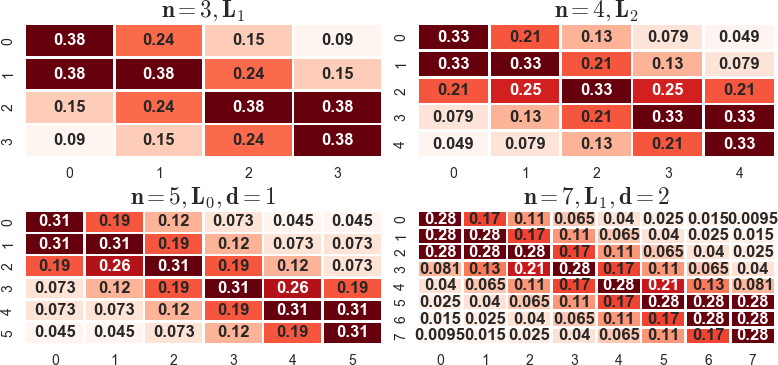}
  \caption{Heatmaps of constrained mechanisms for $\alpha=0.62$ }
  \label{fig:noanomalies}
\end{figure}

\smallskip
Clearly, such results are counter-intuitive and show that blind
optimization of simple objective functions leads to unexpected and
undesirable outcomes.  
To address this, we initiate the study of {\em constrained mechanism design}:
requiring mechanisms to satisfy additional properties ensuring desired
structure in obtained mechanism and avoiding these pathologies.
For example, we define the notion of {\em fairness}, which requires
that the probability of reporting the true input is the same for all
inputs; and {\em weak honesty}, where we require that the probability
of reporting the true input is at least uniform (i.e. at least
$\frac{1}{n+1}$).
These both ultimatel entail that every output is reported with a non-zero probability.
We also consider various monotonicity properties, which preclude big
spikes in probability for responses that are far from the truth. 
In total, we describe seven natural properties that one could demand
of a mechanism.
Our first result is to show how to extend the linear programming framework to incorporate
these properties and eliminate pathological outcomes.

Figure~\ref{fig:noanomalies} shows the heatmap of constrained mechanisms
satisfying all properties.
The anomalies (spikes and gaps) seen in Figure~\ref{fig:anomalies} are
now eliminated.
Recall that optimizing in the unconstrained \2 case returns a trivial
solution that outputs $2$ irrespective of input.
Now, the probability mass in the corresponding constrained mechanism
is more distributed and with probability at least $\frac23$, the mechanism
outputs a value differing from the true answer by at most $1$ for all
inputs.
Similar observations can be made in the other instances. 
We go on to perform a detailed study of constrained mechanism design for count
queries, and show some surprising outcomes:

\eat{
\para{\textbullet~Fully constrained mechanisms minimizing $\0,\1,\2$ are similar}
The mechanisms for satisfying all constraints irrespective of what objective function they are minimizing are similar. This means analyzing properties on just one of the loss functions should give us an approximate idea of utility offered on others. Hence, we focus most of our attention on $\0$ loss function.
}
\para{\textbullet~No blow up in number of mechanisms for $\0$.}
  Given 7 different properties, there are $2^7 = 128$ different
  combinations that could be requested.  Does this mean that there are
  over a hundred distinct constrained mechanisms?
  We show that this is not the case: there are at most four different
  behaviours that can be observed.
  Two behaviors correspond to explicit constructions of mechanisms:
  the (truncated)  geometric mechanism (\gm) proposed in \cite{Ghosh:2009}, which corresponds to the unconstrained optimal solution; and a new ``explicit fair
  mechanism'' (\E) which simultaneously achieves all the properties
  that we introduce.
  In between are two mechanisms which achieve variations of the weak
  honesty property above, which are found by solving an optimization
  problem.
  
\para{\textbullet~No significant loss in utility.}
  The Geometric mechanism obtains the minimal value of the $\0$ loss
  function, for which we give a closed form in terms of the privacy
  parameter $\alpha$.
  However, our most constrained mechanism (the explicit fair
  mechanism,  \E) is only incrementally more expensive: the loss function
  value is higher by a factor of  approximately $1 + \frac{1}{n}$,
  which becomes negligible for even moderate $n$.
  The costs of the other constrained mechanisms are sandwiched in
  between.

\smallskip
Consequently, we conclude that the addition of constraints provides
significant structure to the space of mechanism design, and comes at
very low cost.
Given these observations, one may wonder whether there is any material
difference in behavior between the constrained and unconstrained
mechanisms?
This is indeed the case. 
For example, Figure~\ref{fig:heatmaps} shows a quantitative
difference between \gm and \E for $n=4$ (chosen to make the
results easy to view).
The heatmap shows that \gm concentrates the probability mass on the
two extreme outputs, 0 and $n$, while \E achieves a more balanced
distribution, closer to the leading diagonal (corresponding to a
truthful mechanism).
If we assume a uniform input distribution, \E reports the true input
with probability 0.224, while \gm (which maximizes this quantity)
achieves 0.238, only marginally higher but with a high skew. 
A third mechanism with the weak honesty property, \wm, sits between
the two.

Our experiments further study the implications of using
constrained mechanisms, and compare their empirical behavior on a
mixture of real and synthetic data. 
Differences are most apparent for moderate values of $n$: as $n$
becomes very large, these ``end effects'' become less
significant, and off-the-shelf mechanisms do a good enough job.
Thus, we spend most of our effort studying groups corresponding to a
moderate number of individuals, up to tens.
Arguably, such small groups are most in need of protection, since they
have only a few participants: there is reduced safety in numbers for them.

\para{Outline.}
First, we discuss prior work and introduce our model and define useful
notions for differential privacy in Section~\ref{sec:prelims}.
After defining the Linear Programming framework (Section~\ref{sec:unconstrained}),
in Section~\ref{sec:smallgroup}, we present constraints that can be added to avoid degeneracy. 
We show that additional properties which constrain the output can be
obtained efficiently via solving a constrained optimization problem. 
We also propose an explicit construction of a mechanism which provably
achieves all our proposed properties, and analyze the additional ``cost'' in
terms of various measures of accuracy.
In Section~\ref{sec:expts}, we report on intrinsic
experiments to  study the properties of our mechanisms in a data
agnostic way, and complement this with extrinsic experiments to test
the accuracy on real data.

%% file: model.tex
\section{Preliminaries}
\label{sec:prelims}

\subsection{Model And Definitions}
\label{sec:definitions}  	

Our model captures a group of $n$ participants, each of whom has some private
information which is encoded as a single bit. 
They share their information with a trusted aggregator, whose aim is to
release information about the sum of the values while protecting the privacy of each participant.
Although simple, this question is at the heart of all complex analysis
and modelling, and demands a comprehensive solution. 
We simplify the description of the input to just record the true sum
of values $j$, so we have $0 \leq j \leq n$. 
This captures the case of a count-query over a table 
$D$.
Our goal is to design a {\em randomized mechanism} that, given input
$j$ produces output $i$, subject to certain constraints.  

\begin{definition}[Randomized Mechanism]
\label{def:mechanism}
  A randomized me\-ch\-anism $\mathcal{M}$ is a mapping
$\mathcal{M}: D\Rightarrow R$, where $R = \{0,..,n\} = [n]$
  is the range of the mechanism.
We write $\Pr_{\mathcal{M}}[i|j]$ for the conditional probability
that the output $\mathcal{M}(j)$ (on input $j \in D$) is $i \in R$.
We will drop the subscript $_{\mathcal{M}}$ in context. 
\end{definition}

Our mechanism maps inputs in the range $0$ to $n$ to outputs in the
same range. 
While one could allow a different set of outputs,
it is most natural to restrict to this range.
Consider for example, a downstream analysis step which expects counts
to be integers in the range $[n]$: we should ensure that this
expectation is met by the result of applying mechanisms. 
Rather than attempt to map different outputs to this range, it is more direct
to build mechanisms that cover this output set. 
It is therefore natural to represent $\mathcal{M}$ as an $(n+1) \times
(n+1)$ square
matrix $\mathcal{P}$, where $\mathcal{P}_{i,j} = \Pr[\mathcal{M}(j) =
  i] = \Pr_{\mathcal{M}}[i|j]$.
For brevity, we  abbreviate this probability to $\Pr[i|j]$.
Note that therefore $\mathcal{P}$ is a {\em column stochastic matrix}:
the entries in each column can be interpreted as probabilities, and
sum to 1.

\para{\bf Privacy of a mechanism.}
Differential privacy imposes constraints on the probabilities in our
mechanism.
Specifically, it bounds the ratio of probabilities of seeing the same
output for {\em neighboring inputs}~\cite{Dwork06}.
In our setting, the notion of neighboring is simply that they differ
by (at most) one, which happens when an individual changes their response. 
Hence, applying the definition, we obtain

\begin{definition}[Differentially Private Mechanisms]
\label{def:dp}
  \mbox{Mechanism $\mathcal{M}$ is $\alpha$-differentially private for $\alpha
    \in [0,1]$ if}
\[ \textstyle
 \forall i, j : \alpha \leq
 \frac{\Pr[i|j]}{\Pr[i|j+1]}\leq \frac{1}{\alpha}.
\]
\end{definition}%
\noindent
Here $\alpha$ close to 1 provides a stronger notion of privacy and a
tighter constraint on the probabilities, while $\alpha$ close to zero
relaxes these constraints. 
It is common in differential privacy to write $\alpha =
\exp(-\epsilon) \approx 1-\epsilon$, for some $\epsilon > 0$.
We adopt the $\alpha$ notation for conciseness, and translate results
in terms of $\epsilon$-differential privacy when appropriate.
We say 
a DP constraint is {\em tight} if the relevant
inequality is met with equality. 


\para{\bf Utility of a mechanism.}
The true test of the utility of a mechanism is the accuracy with
which it allows queries to be answered over real data.
However, we aim to design mechanisms prior to their application to
data, and so we seek a suitable function to evaluate their quality.
Since there are many column stochastic matrices that satisfy DP, the problem of finding a mechanism that provides the maximal utility can be framed as an optimization problem. Specifically, we can encode our notion of utility  as a penalty function, where we seek to penalize the mechanism for reporting results that are far from the true  answer.  
\begin{definition}[Objective function value]
\label{def:objective}
  We define the objective function $O_{p,\oplus}(\mech)$ of a mechanism $\mech$ as:

\[
O_{p,\oplus}(\mech) = \oplus_j \sum_{i} w_{j} \Pr[i|j] |i-j|^{p}
\]
\noindent\mbox{where $\oplus$ is an operator like $\sum$ or
  $\operatorname{max}$, and $\sum_{j} w_j=1$.}
\end{definition}

Observe that the weights $w_j$ can be thought of as a prior
distribution on the input values $j$.
Then $O_{p,\sum}(\mech)$ gives the expected error of the mechanism, when
taking its output as the true answer, and $|i-j|^{p}$ penalizes the
extent by which the output was incorrect.
When not otherwise stated, we take $w_j = \frac{1}{n+1}$, i.e. a
uniform prior over the inputs.
Common choices for $p$ in the definition would be $p=2$,
corresponding to a squared error ($\2$ norm), $p=1$, corresponding to an absolute
error ($\1$ norm), and $p=0$, corresponding to the probability of any wrong
answer ($\0$ norm).
In what follows, we devote most of our attention to the case \0.
We argue that this is an important case:
(i) maximizing the probability of reporting the truth is a natural
objective in mechanism design; we aim to ensure that the reported
answer is the maximum likelihood estimator (MLE) for the true answer,
for use in downstream processing
(ii) due to the differential privacy constraints, maximizing the probability of the
true answer has the additional effect of making nearby answers likely,
as our experiments validate.
(iii) our internal study shows that objectives like \1 and \2
often give pathological results, as seen in
Figure~\ref{fig:anomalies}.
Working with \0 gives more robust behavior.
We therefore initiate the study of constrained mechanism design for \0, and give some
initial results for other objectives.
It is convenient to apply a rescaling of the loss function by a factor of
$\frac{n+1}{n}$: this sets the cost of a trivial mechanism to 1
(Definition \ref{def:uniform}).
We refer to this rescaled cost as \0, as this corresponds to a scaled
version of
$O_{0,\sum}$  that sums the probabilities of a wrong answer, and~so
\begin{equation}
  \0(\mech) = \frac{n+1}{n}-\frac{\trace{\mech}}{n}.
  \label{eq:l0cost}
\end{equation}

Abusing notation slightly, we also define the 
objective function, $\Ld{d} = \frac{n+1}{n} \sum_{i,j : |i-j| \geq d}^{n} w_j \Pr[i|j]$ which
computes a rescaled sum of probabilities more than $d$ steps off the main
diagonal, so that $\0 = {\0}_{,0}$. 

%

\subsection{Prior Work and Existing Mechanisms}
\label{sec:prior}
We now review the most relevant existing approaches that
apply in our setting. The model of differential privacy
\cite{dinur:2003,Dwork06,dwork:2004} has received a lot
of attention in the decade since it was christened, 
from a variety of communities including
systems~\cite{shi:2011}, machine learning and signal processing~\cite{sarwate:2013} and data management~\cite{yang:2012}.
For a more thorough overview of the area, there are several detailed surveys~\cite{hardness11,Dwork:Roth:14,VLDB16}.

The most relevant work to our interests is due to Ghosh \etal~\cite{Ghosh:2009} who study the problem of designing 
mechanisms optimizing for expected utility.
Their contributions are to introduce a linear programming formulation
of the problem, and to show that a certain mechanism (denoted \gm) emerges as the
basis of other optimal mechanisms, discussed in more detail below. 
Gupte and Sunararajan proved a similar universality result for ``minimax'' loss
functions and uniform weights $w_j$~\cite{Gupte:2010}.
They provided a simple test for when a given mechanism can be obtained
by first applying \gm and then modifying the result (e.g. by randomly
sampling from a distribution indexed by the observed output from \gm). 
Subsequent work by Brenner and Nissim shows that such ``universally
optimal'' mechanisms are not possible in general for other
computations, such as computing histograms~\cite{Brenner:Nissim:10}. 
Other relevant work studies special cases of differential
privacy.
An important variant is the model of {\em local differential
  privacy} (LDP), where users first perturb their input before passing it to
an (untrusted) aggregator. 
That is, each user applies a mechanism for a group of size $n=1$. 
LDP is used  in Google's Chrome via the
RAPPOR tool to collect browser and system statistics~\cite{rappor},
and in Apple's iOS 10 to collect app usage statistics~\cite{apple}.

\noindent
\mbox{The most relevant existing approaches to us are the following:}

\para{Mechanisms from coin-tossing: Randomized Response.} 
There are many variations of Randomized Response \cite{Chaudhuri:1988}.
A canonical form for the case $n=1$ 
has the user report the true value of their input bit
with probability $p > \frac12$, but report the negation of their input with
probability $1-p$.
%
%
It is immediate that this procedure achieves 
$\alpha$-differential privacy for $\alpha=\frac{1-p}{p}$ (see Definition~\ref{def:dp}). 
%
Due to its simplicity and privacy guarantees, randomized response has
recently found use in a number of systems, such as RAPPOR~\cite{rappor}, which
applies randomized response in conjunction with a Bloom filter to accommodate many possible elements. 
Geng \etal\ in \cite{GKOP:15} give a natural extension of $1$ bit
randomized response to $n$-ary data, which reports its input with
probability $p$, else another output is chosen uniformly.
This gives low utility for count queries. 
\eat{
\para{Staircase Mechanism}
\[ \begin{pmatrix}
p & \frac{1-p}{n} & \frac{1-p}{n} & \cdots &\frac{1-p}{n} \\ 
\frac{1-p}{n} & p & \frac{1-p}{n}& \cdots &\frac{1-p}{n} \\ 
\frac{1-p}{n} & \frac{1-p}{n} & p &\cdots &\frac{1-p}{n} \\
\cdots & \cdots& \cdots& \ddots \\ 
\frac{1-p}{n} & \frac{1-p}{n} & \frac{1-p}{n} &\cdots & p
\end{pmatrix} 
 \]
}

\para{Defining sampling probabilities: Exponential Mechanism.}
McSherry and Talwar \cite{McSherry:2007} proposed the Exponential
Mechanism as a generic approach to designing mechanisms. 
Let $\mathcal{D}$ be the domain of input dataset and $\mathcal{R}$ the range of perturbed responses.
The crux of the exponential mechanism is in designing a {\em quality
  function} $Q:\mathcal{D} \times \mathcal{R} \Rightarrow \mathbb{R}$
so that $Q(d,r)$ measures the desirability of providing output $r$ for
input $d$.
The mechanism is then defined by setting
\begin{equation}
  \textstyle
  \Pr[ r \in \mathcal{R} | d] =  \exp{\left(\frac{\epsilon Q(d,r)}{2s}\right)}
  \Big/ \sum_{r'
      \in \mathcal{R}} \exp{\left(\frac{\epsilon Q(d,r')}{2s}\right)}
  \label{eq:em}
\end{equation}

\noindent
where $s$ captures the amount by which changing an individual's input can alter
the output of $Q$ in the worst case.
It is proved that this mechanism obtains at least
$\exp(-\epsilon)$-differential privacy. 
However, although we can use $Q$ to indicate that some outputs are
more preferred, it is not possible to modify a given $Q$ to directly enforce the
properties that we desire, such as ensuring that the probability of
returning the true output is at least as good as that of a uniform
distribution (``weak honesty'', \eqref{eq:weakhonest}).


\para{Rounding numeric outputs: Laplace and Geometric Mechanisms.}
\label{sec:gm}
Perhaps the best known differentially private mechanism is the Laplace
mechanism, which operates by adding random noise to the true answer
from an appropriately scaled Laplace distribution (a continuous
exponential distribution symmetric around zero).
Note that in order to fit our definition of a mechanism
(Definition~\ref{def:mechanism}), it will be necessary to round and
truncate the output of the mechanism to the range $[n]$.
Here, the Laplace mechanism does not easily fit the requirements.
Instead, the appropriate method is the discrete analog of the Laplace
mechanism, which is the (truncated)
Geometric mechanism, introduced by Ghosh \etal\ \cite{Ghosh:2009}, who
showed that it is the basis for unconstrained mechanisms.

\begin{definition}{Range Restricted Geometric Mechanism
    \protect\cite{Ghosh:2009} (\gm)}
\label{def:gm}
  Let $q$ be the true (unperturbed)
  result of a count query. The \gm responds with $\min(\max(0,
  q+\delta), n)$, where $\delta$ is a noise drawn from a random
  variable $X$ with a double sided geometric distribution,
  $\Pr[X=\delta]=\frac{(1-\alpha)^{|\delta|}}{1+\alpha}$ for
  $\delta \in \mathbb{Z}$.
\end{definition}

That is, \gm
adds noise from two sided geometric distribution to the query result
and remaps all outputs less than $0$ onto $0$ and greater than $n$ to $n$.
Though \gm does not include any zero rows, we observe that each column distribution in \gm has spikes at the extreme values, which tend to distort the true distribution quite
dramatically,  
as the next example shows. 

\begin{example}
  \label{eg:gm}
Consider the case of $n=2$, corresponding to a group of two
individuals, with a moderate setting of the privacy
parameter $\alpha=\frac{9}{10}$.
For an input of $1$ (i.e. one user has a 1, and the other has a 0), 
we obtain
that the probability of seeing an output of 0 is 
$\approx 0.47$, and the same for an output of 2. 
Meanwhile,
the probably of reporting the true output is $\approx 0.05$
--- in other words, the chance of seeing the
true answer is eighteen times lower than seeing an incorrect answer.
Meanwhile,
if the input is 0, then output 0 is returned with probability 
$\approx 0.53$:
so the mechanism is much more likely
to report the true answer when it is 0 than when it is 1. 
As we increase the privacy parameter $\alpha$ closer to 1 (more
privacy), the probability of outputs other than 0 and $n$ approaches 0. 
\end{example}

%% file: expts_t.tex
\begin{figure*} 
\centering
  \subfigure[Varying group size]{
    \includegraphics[width=0.45\textwidth]%
                    {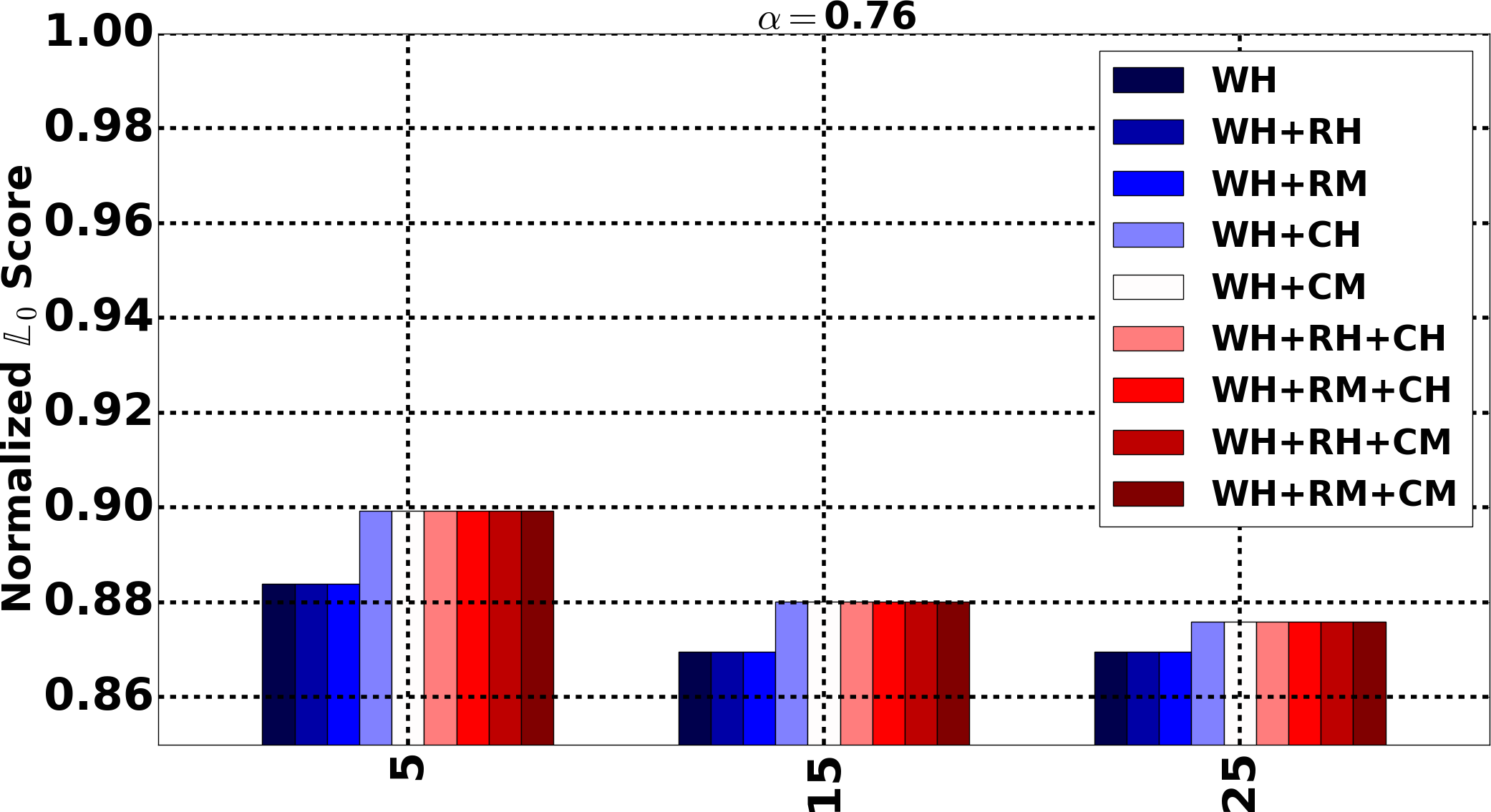}
\label{fig:wmn}
  }%
  \subfigure[Varying $\alpha$]{
    \includegraphics[width=0.45\textwidth]%
    {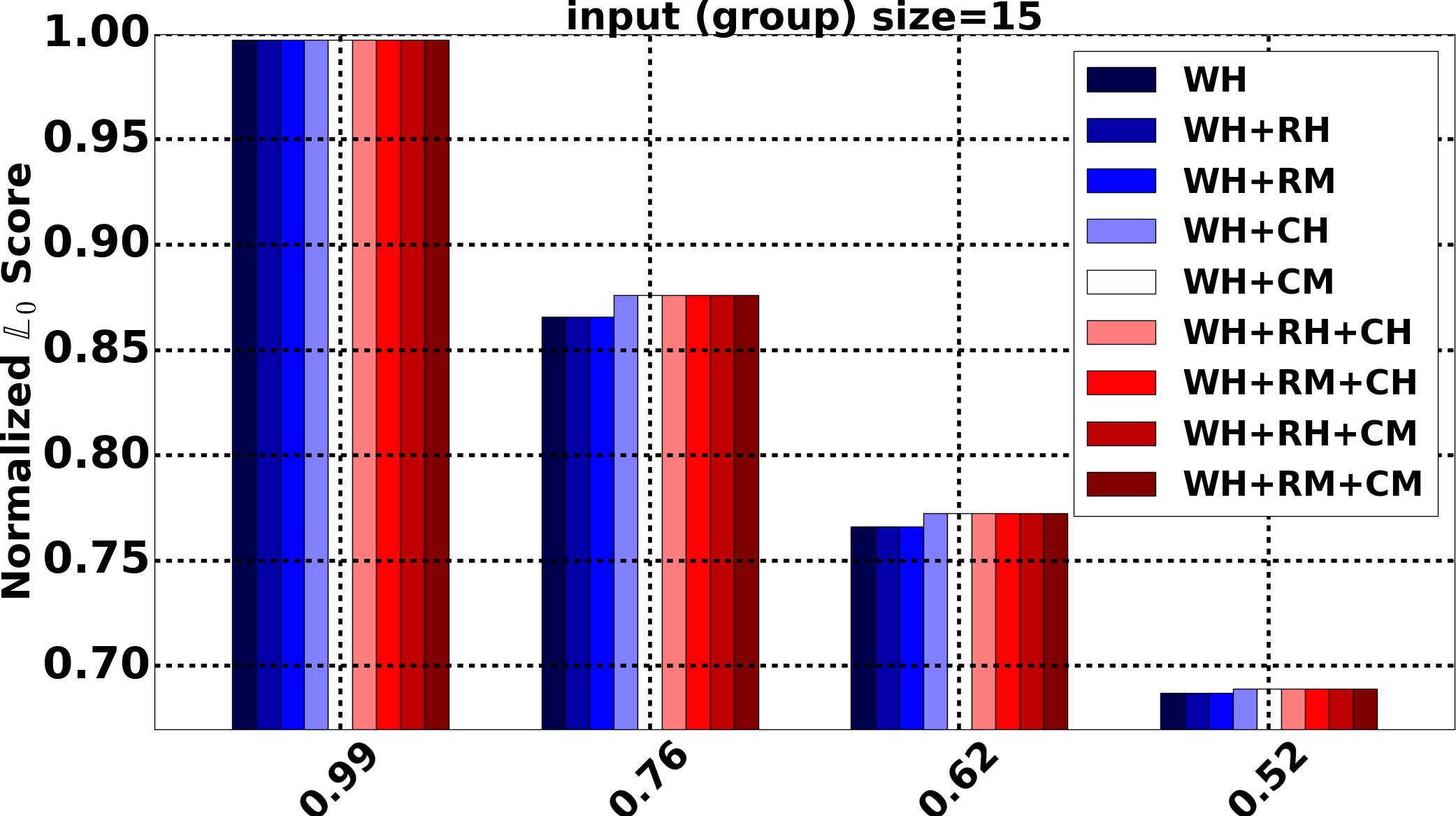}
\label{fig:wma}
  }
\caption{Combinations Of Properties with Weak Honesty}
\label{fig:weakhonesty}
\end{figure*}

\begin{figure*} 
  \subfigure[$\alpha=\frac23$]{
     \includegraphics[width=0.33\textwidth]{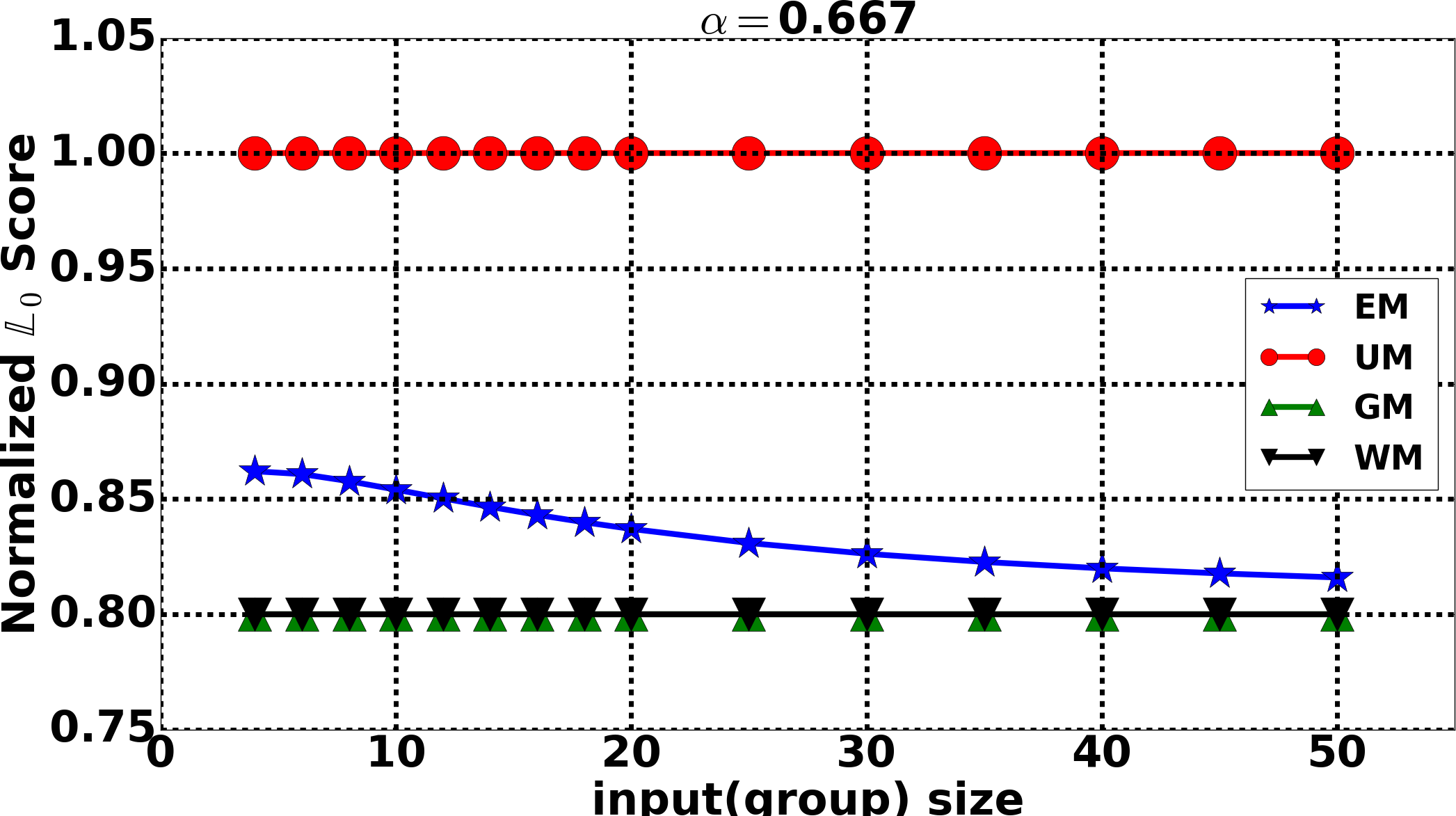}
     \label{fig:l0wm66}
  }%
  \subfigure[$\alpha=\frac{10}{11}$]{
     \includegraphics[width=0.33\textwidth]{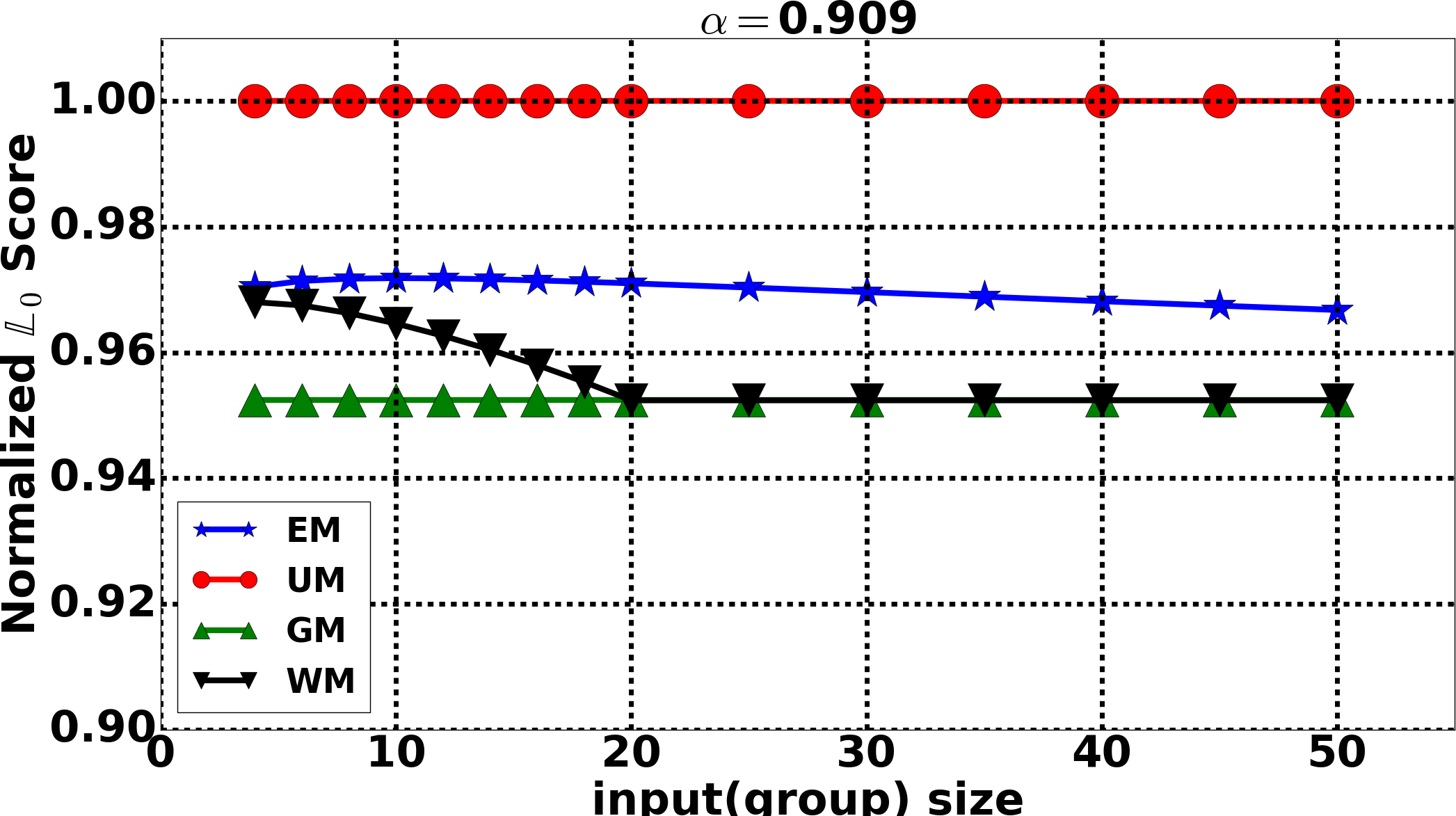}
     \label{fig:l0wm90}
  }%
  \subfigure[$\alpha=\frac{99}{100}$]{
     \includegraphics[width=0.33\textwidth]{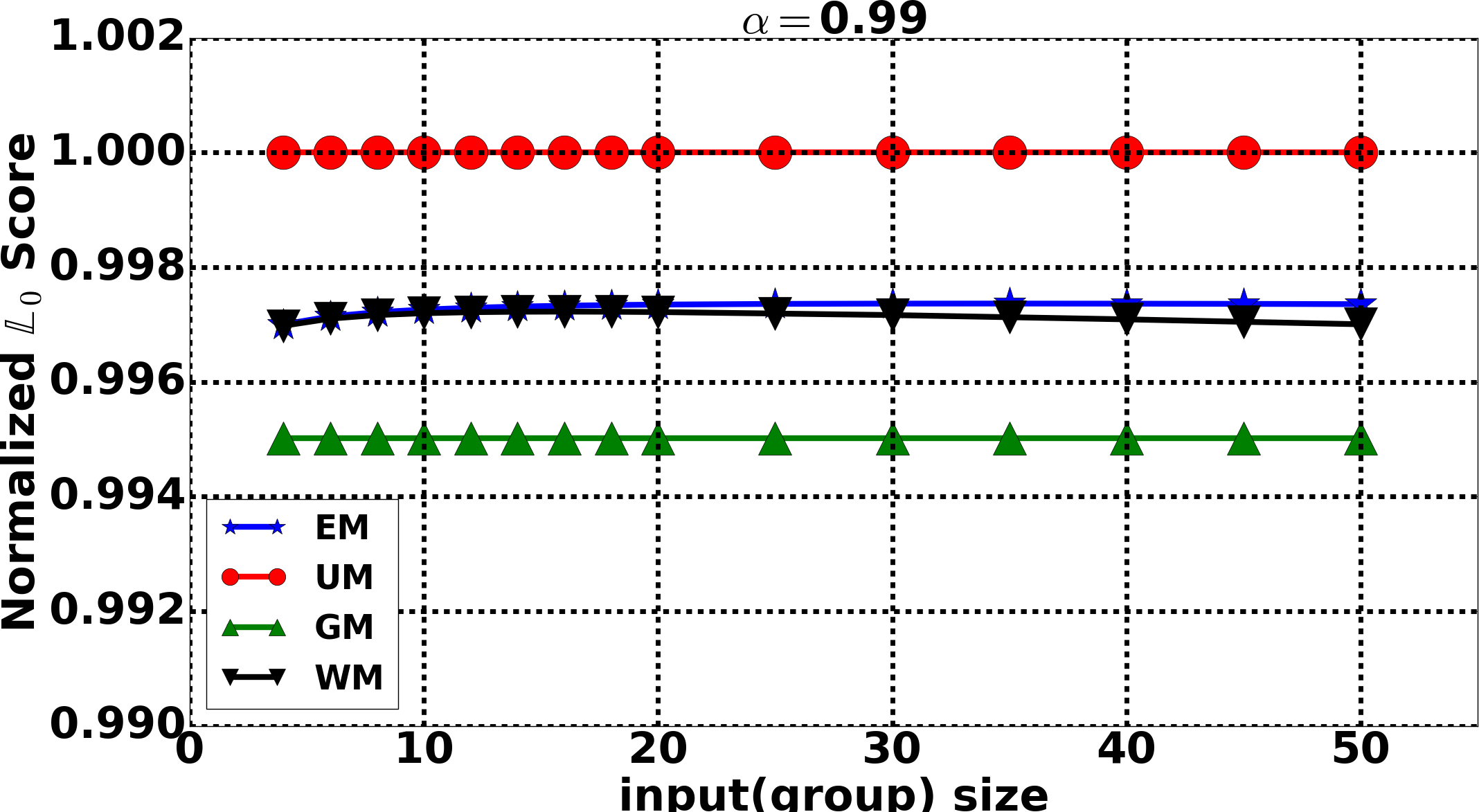}
     \label{fig:l0wm99}
  }
\caption{Final Groups Of Mechanisms with Distinct Behaviors}
\label{fig:l0wm}
\end{figure*}

\section{Experimental Study}
\label{sec:expts}

In Section~\ref{sec:exptswm}, we substantiate our earlier claims about
properties of mechanisms satisfying weak honesty (but not fairness). 
In what follows, we look at
other measures of utility of these mechanisms, to understand their
robustness. 


\para{Default Experimental Settings.} All experiments in this work were implemented in Python, making use of the standard library NumPy to handle
the linear algebraic calculations, and PyLPSolve \cite{pylpsolve} to solve the generated LPs. Evaluation was made on a commodity machine running Linux. We omit detailed timing measurements, as the time to solve the LPs generated was negligible (sub-second).

\begin{figure*}[t]
\centering
  \subfigure[Estimating young population]{
        {\includegraphics[clip,width=0.33\textwidth]{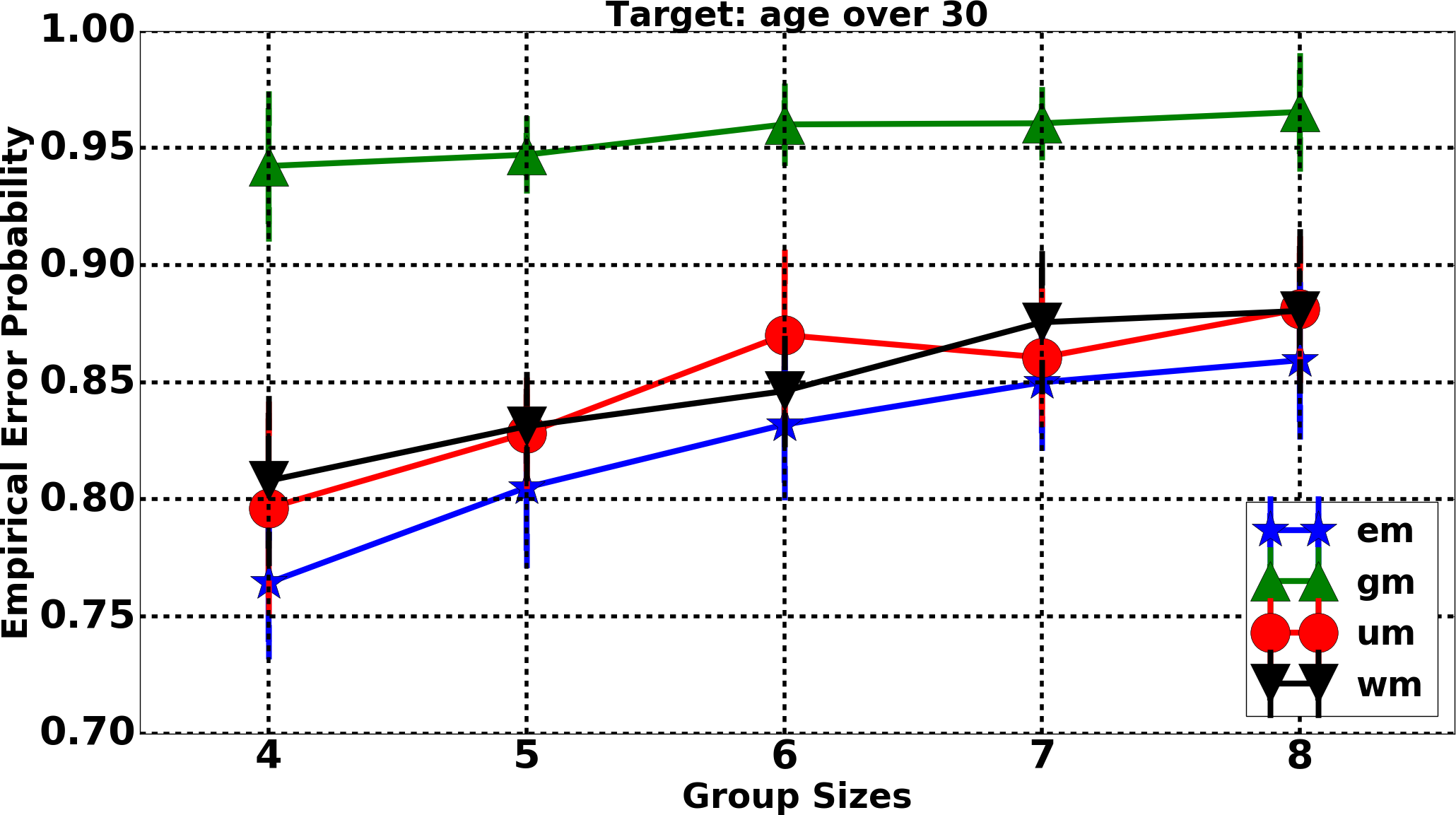}}
    \label{fig:adulteducation}
  }%
  \subfigure[Estimating gender balance]{
     \includegraphics[clip,width=0.33\textwidth]{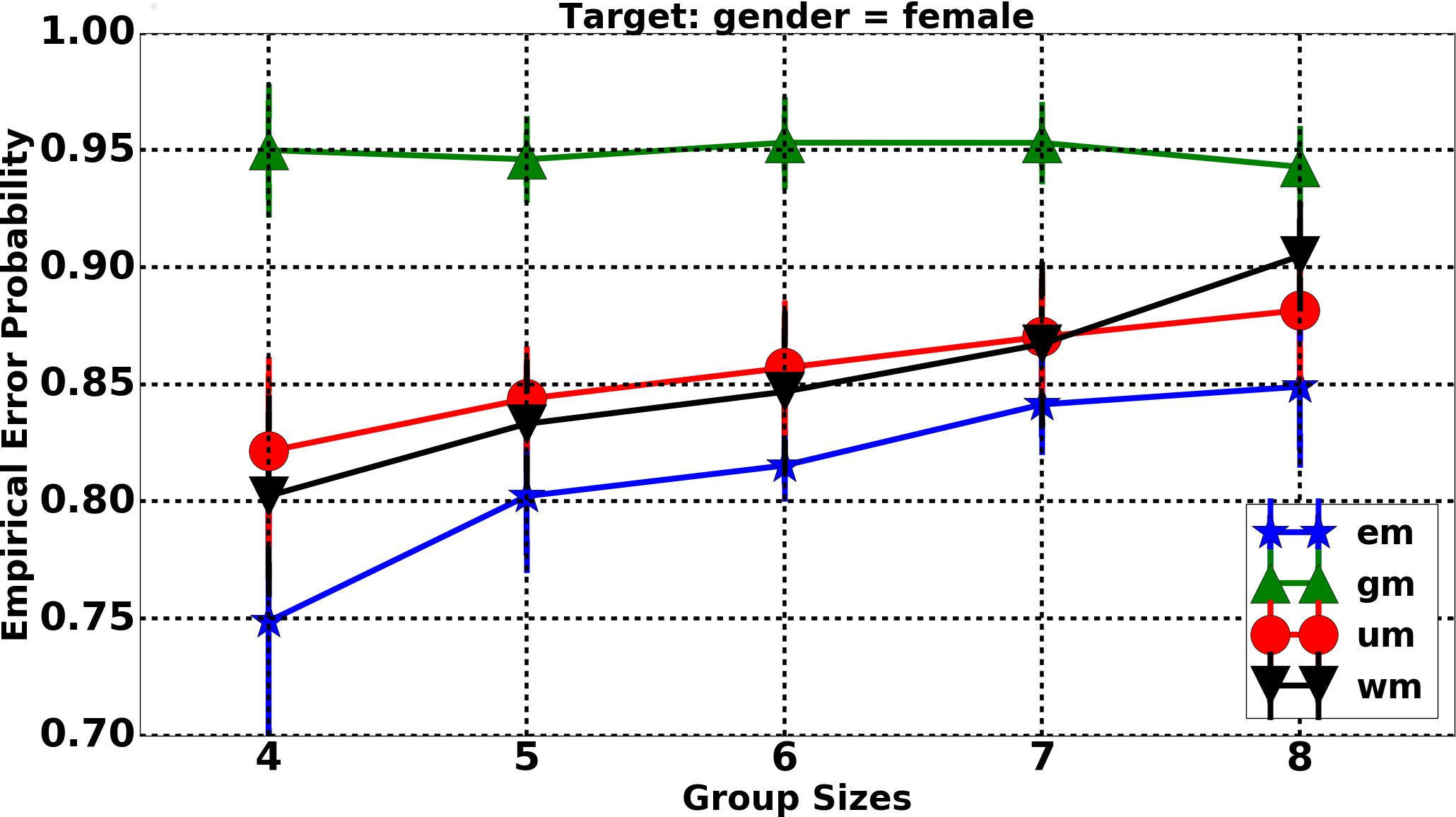}
    \label{fig:adultsex}
  }%
  \subfigure[Estimating income level]{
        \includegraphics[clip,width=0.33\textwidth]{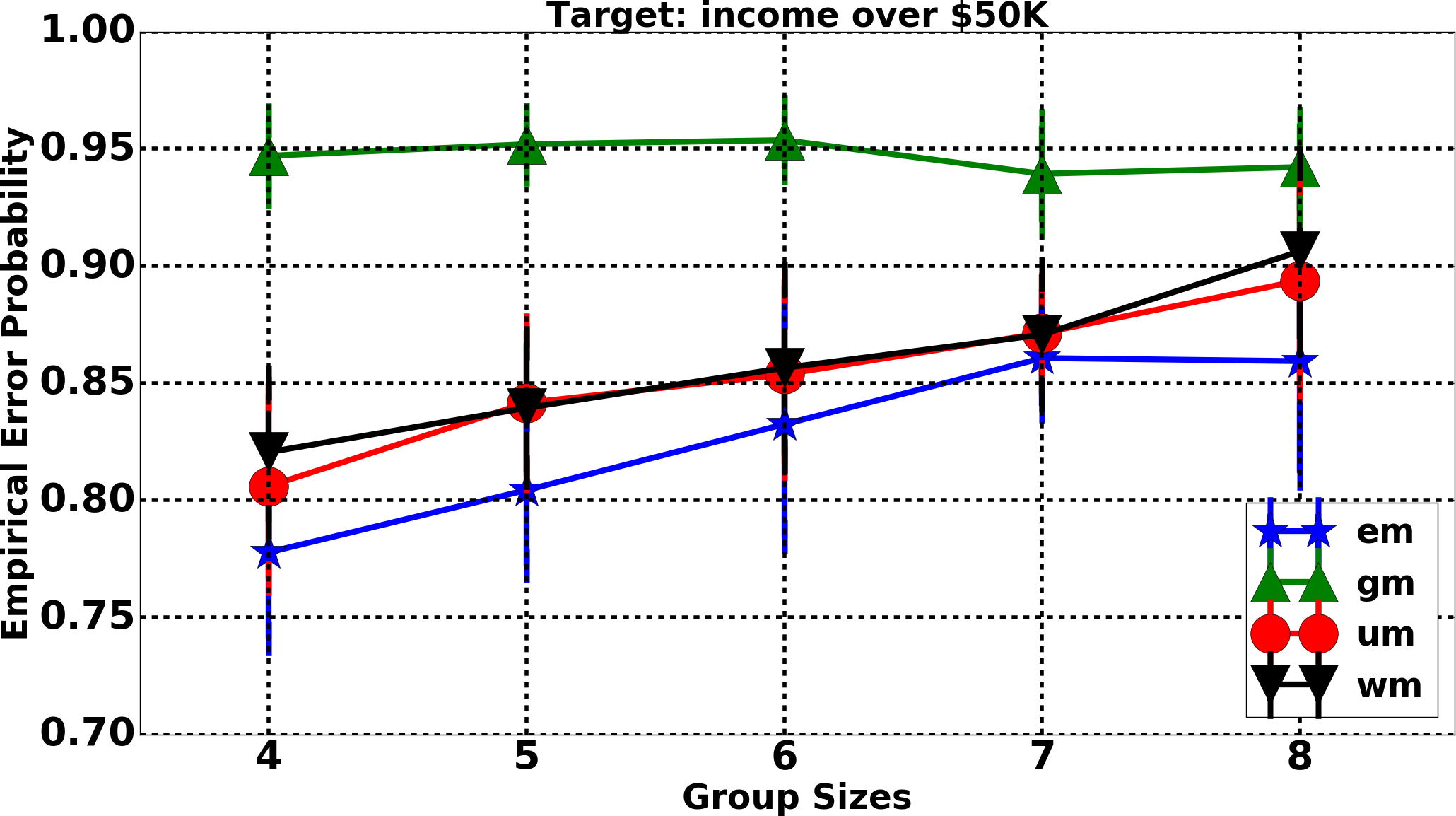}

    \label{fig:adultrich}
  }
\caption{Empirical Error Probability on Adult Dataset for $\alpha=0.9$}
 \label{fig:l0_scores_adult}
\end{figure*}

\begin{figure*}[t]
\centering
  \includegraphics[width=0.85\textwidth]{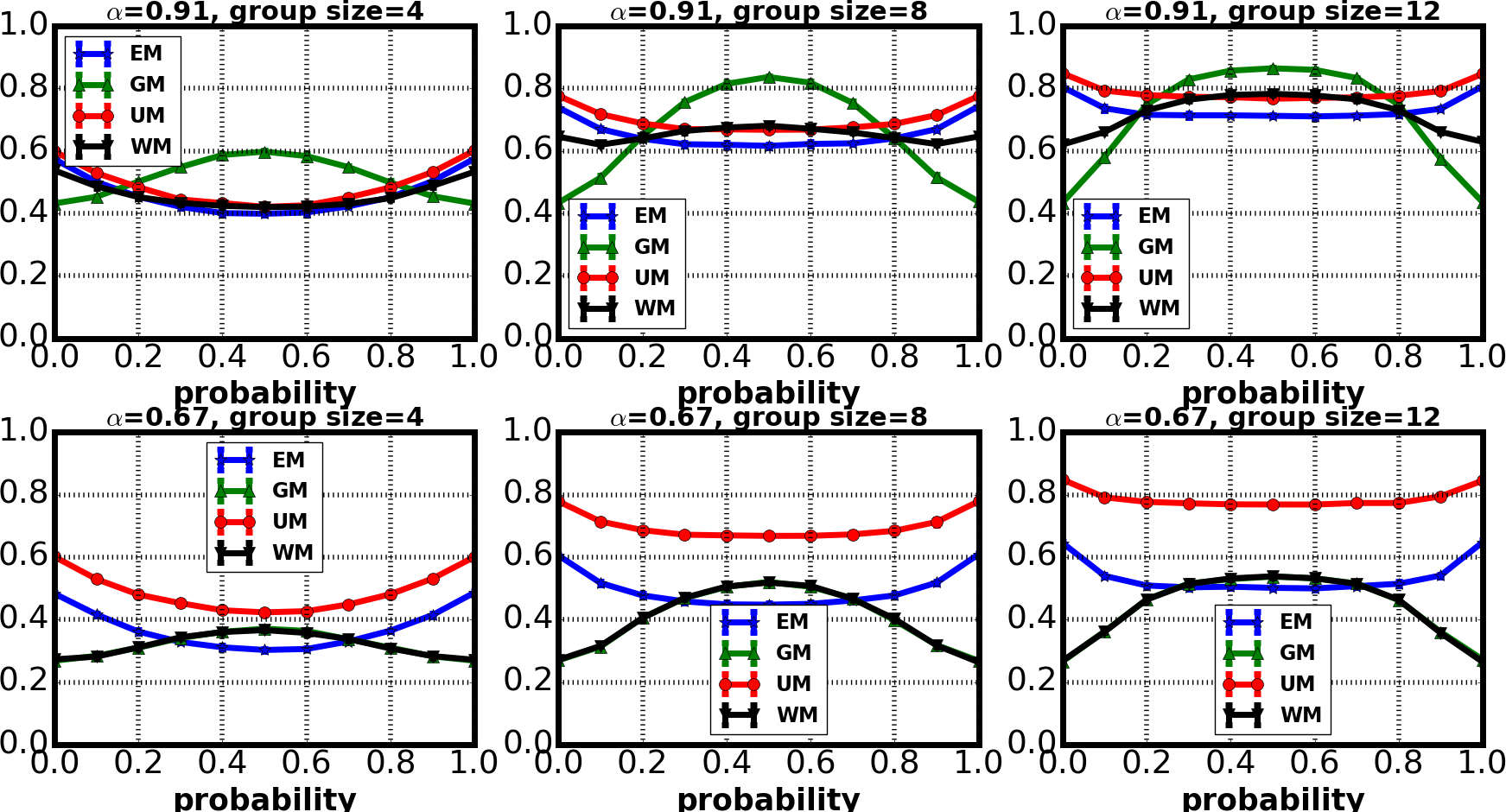}	
  \caption{\Ld{1} score for Binomial data, for $n=\{4, 8, 12\}$ and $\alpha= \{0.91,0.67\}$}
  \label{fig:l_0_1}
\end{figure*}

\para{Experimental Setting.} We considered a variety of settings of parameter
$\alpha$ (typical values chosen are $\{\frac12, \frac23, \frac{10}{11}, \frac{99}{100}\}$
and group size $n$ (ranging from 2 up to hundreds).

\subsection{\0 Objective Function}
\label{sec:exptswm}

Our first experiment analyzes the effect of
weak honesty combined with other properties drawn from $\{$CH,
CM, RH, RM$\}$,
including the empty set. 
There are 9 meaningful combinations of properties to ask for, which we
write as  $\{\emptyset,$ RH, RM, CH, CM, RH+CH, RH+CM, RM+CH,
RM+CM$\}$ --- other combinations reduce to these, since RM implies RH,
and CM implies CH. 

As discussed in Section \ref{sec:comparing}, there
are cases when the solution found by solving the LP has cost
$\frac{2\alpha}{1+\alpha}$ and is identical to \gm: these are when
$n \ge \frac{2\alpha}{1-\alpha}$ and only row-wise properties are
requested, consistent with Figure \ref{fig:flowchart}.
This is borne out in Figure \ref{fig:weakhonesty}:
we see that when WH alone is requested, or in combination with only row
properties (RH or RM) we get a lower \0 value than when any column
properties (CH or CM) are requested.
Figure~\ref{fig:wmn} shows the case for different values of $n$.
When $n > \frac{2\alpha}{1-\alpha}$, which is 6.33 in this example (where $\alpha=0.76$),
the cost of WH alone is $\frac{2\alpha}{1+\alpha} = 0.864$, the cost
of \gm.
For large $\alpha$ (Figure~\ref{fig:wma}), the cost of all
combinations of WH are the same, and identical to the cost of \E;
as $\alpha$ is decreased, we see two behaviors, where the lower \0 cost
is that of \gm. 
We confirmed this behavior for a wide range of $n$ and $\alpha$
values. 
From now on, we use \wm to refer to the mechanism with WH, RM and CM properties. 


The relationship between the \0 scores for the three mechanisms is
further clarified in Figure~\ref{fig:l0wm}. 
The plots show the \0 scores of \gm, \wm, \E and \UM for different
values of $\alpha$.
In Figure \ref{fig:l0wm66}, $\alpha=\frac23$ so the threshold
$\frac{2\alpha}{1-\alpha} = 4$.
Then \gm satisfies WH for the whole range of $n$ values shown, so \wm
converges on \gm, while \E has a higher (but decreasing) cost.
For Figure \ref{fig:l0wm90}, $\alpha=10/11$ so the threshold is 20.
Indeed, we see that the cost of \wm converges with \gm at $n=20$.
Last, in Figure \ref{fig:l0wm99}, the threshold of 198 is far
above the range of $n$ values shown, so \wm does not converge on \gm
here.
Rather, for this high value of $\alpha$, the $y$ value
for \E is above $\frac{1}{n+1}$ for all $n$: so in this case \E has
weak honesty, and the cost of \wm remains the same as that of the
optimal fair \E. 


\begin{figure}[t]
 \includegraphics[width=0.5\textwidth]{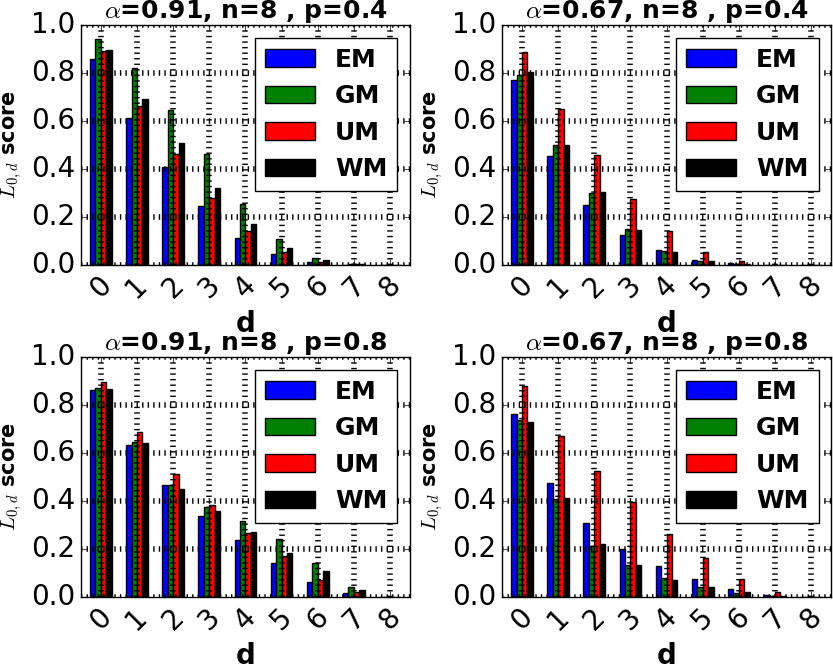} 
  \caption{Histograms of \Ld{d} scores for Binomial data}    
 \label{fig:histogram_plots}
  \end{figure}

\begin{figure*}
\centering
  \includegraphics[width=0.9\textwidth]{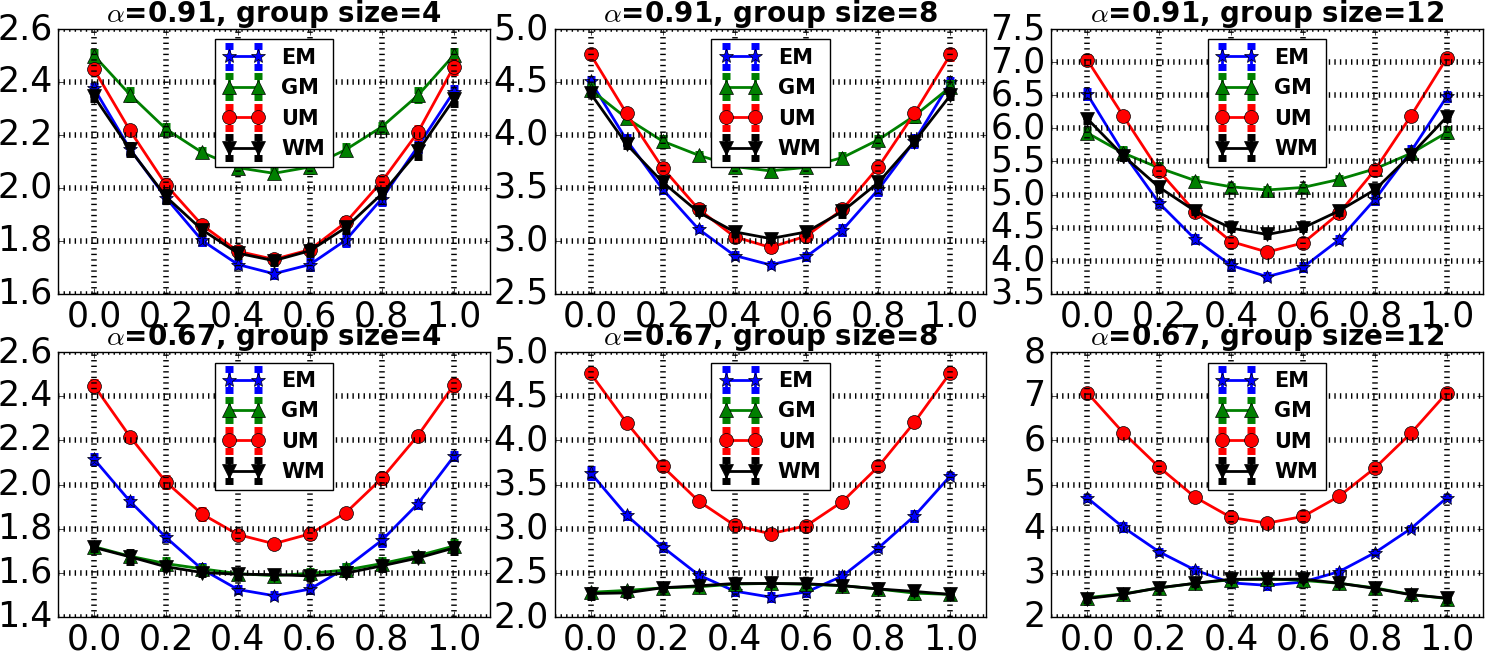}
  
  \caption{Root Mean Square Error Plots For Binomial data}
  \label{fig:rmse}
\end{figure*} 

\subsection{Experiments On Real Data}

We make use of the UCI Adult dataset, a workhorse for privacy
experiments \cite{Blake:1998}.
Our instance of the dataset contains demographic information on 
32K adults with 15 columns listing age, job type, education,
relationship status, gender, and (binary) income level.
We created three binary targets, treated as sensitive:
income level (high/low), gender (male/female), and young (age over/under $30$).
To form small groups, we gathered the rows (corresponding to individuals)
arbitrarily into groups of a desired size.

Figure~\ref{fig:l0_scores_adult} shows results for the \0 objective,
that is, where we focus on the fraction of times the mechanism reports
an incorrect answer, as a function of group size.
Specifically, we count the number of groups whose noisy count for each target
attribute is not equal to their true count.
We expect this quantity to be fairly high, as it measures how often
our mechanism is honest, i.e. returns the true input. 
Other experiments (not shown) computed the corresponding probability
for returning an answer that is close to the true one, e.g. off by at
most one, and showed similar patterns. 
The plot includes error bars from $50$ repetitions of this process to show
$1$ standard error.

Observe first that the performance of \UM is essentially independent
of the input data: the chance of it picking the correct answer is
always $1 - \frac{1}{n+1}$ for a group of size $n$, and indeed we see this
behavior (up to random variation).
We would hope that our optimized mechanisms can outperform this trivial method.
Perhaps surprisingly, on this data \gm does appreciably worse.
This highlights the limitations of \gm.
In this data, the common inputs are around the middle of the group
size (i.e. typically close to $n/2$).
It is on these inputs that \gm does poorly, and only does well for
inputs that are 0 or $n$, which happen to be rare in this dataset
(in other words, the data distribution does not match the prior for
which \gm is optimal).
The condition of weak honesty is not sufficient to improve
significantly over random guessing: for this data, we see that \wm
tracks \UM quite closely.
It is only the most constrained mechanism that fares better on this
evaluation metric for this data: \E which achieves fairness gives the
best probability of returning the unperturbed input. 

In corresponding experiments with higher values of $\alpha$ in the
range $0.9$ to $0.99$, corresponding to the strongest privacy
guarantees adopted in prior work on differential privacy,
there is not much to choose between \E and
\wm, and it gets even harder to show substantial improvement over
uniform guessing. 
In order to understand the behaviors of the mechanisms further,
we next consider synthetic data,
where we can directly control the data skewness within groups.

\subsection{Experiments On Synthetic Data}

In our experiments with synthetic data we generate a population of
$10,000$ individuals with a private bit and divide them into small
groups of the same size, $n$.
Each individual has the same probability $p$ of having their bit be
one, so the distribution within each group is Binomial. 
Hence, the expected count for each group is $pn$. 
Our experiments vary the parameters $p$, $n$ and $\alpha$. 


\para{\Ld{1} Error.}
Our experiments so far have uesd the target objective function \0 to
evaluate the quality of the mechanism.
This is sufficient to distinguish the different mechanisms, but all
mechanisms achieve a score which can be close to 1, obtained by
uniform guessing.
To better demonstrate the usefulness of the obtained mechanisms, we
use other functions to evaluate their accuracy. 
Figure~\ref{fig:l_0_1}
uses the related measure of $\mathbb{L}_{0,1}$ i.e. the fraction of
groups which output a value differing from their true answer by more than $1$,
as we vary data distribution (determined by $p$), group size $n$, and
privacy parameter $\alpha$.
We stress that though we use $\Ld{1}$ for evaluation,
we continue to use mechanisms designed for minimizing the $\0$ error.
Each subplot in the figure represents a configuration of $\langle
\alpha,n \rangle$, describing how $\Ld{1}$ error changes with input
distribution parameter $p$.
Each experiment is repeated $30$ times and we observe that the
results have very small variance.

It is apparent that the shape of the input distribution has a pronounced
effect on the quality of the output.
We confirm that \gm can do well when the input is very biased ($p$
close to 0 or 1), which generates more instances with extreme input
values. 
However, when the input is more spread across the input space, the
more constrained mechanisms consistently give better results. 
For higher $\alpha$, the constrained methods have similar behavior,
and improve only slightly over \UM (while \gm is often worse than
uniform). 
Enforcing fairness tends to make \E less sensitive to the input
distribution, except when the input is an extreme value (0 or $n$). 
When $\alpha$ is lower (second row), the overall scale of error
decreases and \wm and \gm converge, as noted previously. 

    
\para{$\Ld{d}$ Error.}
In the previous experiment, we fixed $d=1$ and evaluated our
mechanisms for variety of input distributions. Next we vary $d$ while
holding input size and input distributions and compute $\Ld{d}$ error.
Figure~\ref{fig:histogram_plots} plots the fraction of population
reporting a value that is more than $d$ steps away from the true
answer for various $d$ values with $n=8$.
This captures the probability mass in the tail of each mechanism.  

In the top row, we use a more proportionate input distribution.
Here, \E outperforms all other mechanisms, sometimes by a substantial
fraction. 
Interestingly, the margin between \E and \gm only increases with
larger $d$.
Once again we see that for higher $\alpha$ values, use of \gm can
yield accuracy worse than mere random guessing.
For lower $\alpha$'s \gm's accuracy increases dramatically but still remains worse than \E 's.
  
In the bottom row, the input distribution is more skewed, which tends
to favor \gm.
However, \E does not do substantially worse than \gm even for this
biased input distribution. 
The intermediate mechanism found by \wm tends to fall between \gm and
\E.
We observed similar behavior for other values of $n$. 

\para{Root Mean Square Error (RMSE).}
Our final set of experiments compute the RMSE error (a measure of variance and bias of a mechanism) of estimates from
small groups. 
Note that none of our mechanisms are designed to optimize this metric,
but we can nevertheless use it as a measure of the overall spread of error. 
Figure~\ref{fig:rmse} shows plots with error bars showing one standard
deviation from $30$ repetitions. 

As seen in previous experiments, a more symmetric input distribution
($p$ closer to 0.5) tends to be easier for most mechanisms --- although
we see cases where \gm finds this more difficult. 
Increasing the group size increases the RMSE, as there is a wider
range of possible outputs, and the constraints ensure that there is
some probability of producing each possible.
Yet again, we see that increasing $\alpha$ tends to make \gm less competitive and find many cases where \gm is worse than random guessing (\UM).
The interesting case may be for fairly high privacy requirements ($\alpha = 0.91$),
where we observe that \E tends to give lower error across all group
sizes and input distributions.

%% file: appendix.tex
\appendix
\allowdisplaybreaks

\begin{IEEEproof}[Proof of Theorem~\ref{thm:symm_mechanism}]
  Our construction to achieve symmetry is simple.
  Define a matrix $M^{S}$ from $M$ as $(M^{S})_{i,j} = M_{n-i,n-j}$.
  Then set $M^{*} = \frac12(M + M^{S})$.
We first observe that $M^*$ is indeed symmetric, since it is equal to 

\centerline{$
  \frac12(M_{i,j} + M_{n-i, n-j}) 
  =   \frac12(M_{n-i, n-j} + M_{n - (n - i), n - (n -  j)}) 
  =  M^*_{n-i, n-j}
$}
\noindent
as required by \eqref{eq:symmetry}.
  The (\0) objective function value is unchanged since (invoking~\eqref{eq:l0cost})
  \[ \textstyle \trace(M^*) = \frac12(\trace(M) + \trace(M^{S})) = \trace(M)\]
  For the other diagonal properties (fairness and weak honesty),
  it is immediate that if either of these properties are satisfied by
  $M$, then they are also satisfied by $M^*$. 
%
  We prove the claim for row properties; the case for column
  properties is symmetric.

\noindent
  (i) Differential privacy: if we have $\alpha \leq
  {M_{i,j}}/{M_{i,j+1}} \leq 1/\alpha$ for all $i,j$, then this
  also holds for $M^{S}_{i,j}/M^{S}_{i,j+1}$.
  Summing both in\-eq\-ualities,
and using that
  $\min{(\frac{a}{b},\frac{c}{d})} \leq \frac{a+c}{b+d} \leq \max{(\frac{a}{b},\frac{c}{d})}$, 
this holds for $M^*$, hence $M^*$
  satisfies differential privacy. 

\noindent
  (ii) Row monotonicity:
  consider a pair $i, j$ with $1 \leq i \leq j$.
  Then we have $M_{j,i-1} \leq M_{j,i}$ (from~\eqref{eq:rowmonotone}).
  It is also the case that
  $n-j \leq n-i < n$, which means that
  $M_{n-j,n-i+1} \leq M_{n-j,n-i}$ (also
  from~\eqref{eq:rowmonotone}).
  Then $M^{S}_{j,i-1} \leq M^{S}_{j,i}$.
  Combining these two inequalities, we have that $M^*_{j,i-1} \leq
  M^*_{j,i}$.  

  (iii) Row honesty: if $\forall i,j.  M_{i,i} \geq M_{i,j}$, then
  $M^{S}_{i,i} \geq M^{S}_{i,j}$ also.
  Summing both in\-eq\-ualities, we obtain $M^*_{i,i} \geq M^*_{i,j}$
  as required. 
\end{IEEEproof}